\documentclass[twocolumn,10pt]{IEEEtran}

\usepackage{amsmath,amssymb,amsthm}
\usepackage{graphicx}
\usepackage{xcolor}
\usepackage[acronym]{glossaries}

\newacronym{5g}{5G}{fifth-generation}
\newacronym{6g}{6G}{sixth-generation}
\newacronym{adc}{ADC}{analog-to-digital converter}
\newacronym{dac}{DAC}{digital-to-analog converter}
\newacronym{dma}{DMA}{dynamic metasurface antenna}
\newacronym{em}{EM}{electromagnetic}
\newacronym{lmmse}{LMMSE}{linear minimum mean square error}
\newacronym{milac}{MiLAC}{microwave linear analog computer}
\newacronym{mimo}{MIMO}{multiple-input multiple-output}
\newacronym{miso}{MISO}{multiple-input single-output}
\newacronym{mrt}{MRT}{maximum ratio transmission}
\newacronym{rf}{RF}{radio frequency}
\newacronym{ris}{RIS}{reconfigurable intelligent surface}
\newacronym{upa}{UPA}{uniform planar array}

\newtheorem{proposition}{Proposition}

\begin{document}
\bstctlcite{BSTcontrol}

\title{Physics-Compliant Modeling and\\Optimization of MIMO Systems Aided by\\Microwave Linear Analog Computers}

\author{Matteo~Nerini,~\IEEEmembership{Member,~IEEE},
        Bruno~Clerckx,~\IEEEmembership{Fellow,~IEEE}

\thanks{This work has been supported in part by UKRI under Grant EP/Y004086/1, EP/X040569/1, EP/Y037197/1, EP/X04047X/1, EP/Y037243/1.}
\thanks{Matteo Nerini and Bruno Clerckx are with the Department of Electrical and Electronic Engineering, Imperial College London, SW7 2AZ London, U.K. (e-mail: m.nerini20@imperial.ac.uk; b.clerckx@imperial.ac.uk).}}

\maketitle

\begin{abstract}
\Gls{milac} has emerged as a promising architecture for implementing linear \gls{mimo} processing in the analog domain, with \gls{rf} signals.
Existing studies on \gls{milac}-aided communications rely on idealized channel models and neglect antenna mutual coupling.
However, since \gls{milac} performs processing at \gls{rf}, mutual coupling becomes critical and alters the implemented operation, not only the channel characteristics.
In this paper, we develop a physics-compliant model for \gls{milac}-aided \gls{mimo} systems accounting for mutual coupling with multiport network theory.
We derive end-to-end system models for scenarios with \glspl{milac} at the transmitter, the receiver, or both, showing how mutual coupling impacts the linear transformation implemented by the \glspl{milac}.
Furthermore, we formulate and solve a mutual coupling aware \gls{milac} optimization problem, deriving a closed-form globally optimal solution that maximizes the received signal power.
We establish the fundamental performance limits of \gls{milac} with mutual coupling, and derive three analytical results.
First, mutual coupling is beneficial in \gls{milac}-aided systems, on average.
Second, with mutual coupling, \gls{milac} performs as digital architectures equipped with a matching network, while having fewer \gls{rf} chains.
Third, with mutual coupling, \gls{milac} always outperforms digital architectures with no matching network.
Numerical simulations confirm our theoretical findings.
\end{abstract}


\section{Introduction}

The unprecedented growth of mobile data traffic and the requirements anticipated for \gls{6g} wireless networks demand transformative shifts in wireless system design \cite{gio20}.
Among the key technologies enabling these capabilities, gigantic \gls{mimo} has emerged, exploiting large antenna arrays to increase spectral efficiency and spatial multiplexing gains \cite{lar14}.
Whereas \gls{5g} systems employ massive \gls{mimo} arrays with tens of antennas, \gls{6g} visions push toward gigantic \gls{mimo}, involving hundreds or even thousands of antennas \cite{bjo25}.
However, scaling conventional digital \gls{mimo} architectures to such dimensions introduces prohibitive hardware costs and computational burdens.
Each antenna typically requires a dedicated \gls{rf} chain with high-resolution \glspl{adc} and \glspl{dac}, and real-time digital signal processing whose complexity scales rapidly with array size.

To address these challenges, researchers have explored a range of alternative architectures that reduce reliance on digital processing.
Hybrid analog-digital beamforming distributes signal processing between baseband digital and \gls{rf} analog domains, lowering the number of \gls{rf} chains while retaining performance \cite{soh16,mol17}.
Similarly, technologies such as \glspl{ris} and \glspl{dma} steer or shape \gls{em} waves in the analog domain \cite{wu21,shl19}.
Parallel to these developments, there has been a revival of analog computing paradigms for communications, motivated by their potential for ultra-fast and highly parallel computation.
The resurgence is driven by advances in analog hardware, including metasurfaces and resistive memory arrays that allow certain operations to be mapped directly onto physical analog structures \cite{sil14,iel18}.

In this context, \gls{milac} has recently been proposed as a class of analog computers capable of processing microwave signals directly in the analog domain \cite{ner25-1}.
A \gls{milac} is a multiport microwave network made of (possibly tunable) impedance elements that linearly map input signals to outputs, effectively implementing linear transformations and beamforming tasks without extensive digital processing.
Early studies showed that a \gls{milac} can compute \gls{lmmse} estimators and matrix inversions with computational complexity scaling quadratically rather than cubically, enabling orders-of-magnitude reductions compared to digital computing \cite{ner25-1}.
Building on this foundation, MiLAC-aided beamforming has been introduced as a way of performing precoding and combining entirely at \gls{rf}, in the analog domain \cite{ner25-2}.
MiLAC-aided beamforming can achieve flexibility and performance comparable to digital beamforming while requiring significantly fewer \gls{rf} chains, lower \glspl{adc}/\glspl{dac} resolution, and reduced digital processing.

Subsequent work has further studied \gls{milac}-aided communications.
The capacity of \gls{milac}-aided \gls{mimo} systems has been characterized under the constraints of lossless and reciprocal \gls{rf} components, and analytically shown to match that of digital beamforming systems \cite{ner25-3}.
Efficient \gls{milac} architectures with reduced circuit complexity have been developed by using graph theory, addressing practical implementation challenges associated with fully-connected microwave networks \cite{ner25-4}.
The performance limits of \gls{milac} in multi-user systems have been investigated \cite{fan26,wu26}.
While a lossless and reciprocal \gls{milac} cannot exactly match the performance of digital beamforming \cite{fan26}, hybrid \gls{milac}-digital beamforming can achieve maximum performance with only $K$ \gls{rf} chains, with $K$ being the number of users \cite{wu26}, in contrast to hybrid analog-digital beamforming requiring $2K$ \gls{rf} chains \cite{soh16}.
Moreover, the problem of low-complexity channel estimation in \gls{milac}-aided systems has been addressed in \cite{zha26}.

Existing work on \gls{milac}-aided communications has largely relied on idealized channel models \cite{ner25-1}-\cite{zha26}.
In particular, prior studies assume perfectly matched antennas and neglect antenna mutual coupling, which is particularly strong in tightly packed antenna arrays.
While these assumptions enable analytical tractability and provide valuable insights into the fundamental potential of \gls{milac}, they neglect \gls{em} non-idealities, such as mutual coupling, that inevitably arise in practical implementations of multi-antenna devices.

Mutual coupling between antenna elements arises when the \gls{em} signal radiated by an antenna induces currents on neighboring antennas, thereby changing the \gls{em} signals they radiate.
These effects are known from classical \gls{mimo} theory to influence the channel characteristics, introducing correlation between antenna ports and affecting the achievable performance.
Mutual coupling has been extensively studied and modeled in conventional digital \gls{mimo} systems using multiport network theory, either through impedance parameters \cite{cle07,ivr10,ivr14}, or scattering parameters \cite{wal04-1,wal04-2}.
Recent works also explored the role of mutual coupling in \gls{ris} \cite{gra21,she22,abr24,ner24} and \gls{dma} \cite{wil22}.
However, the role of mutual coupling in \gls{milac}-aided systems is fundamentally different and potentially more critical.
In \gls{milac}-aided systems, beamforming is performed directly in the \gls{em} domain through a reconfigurable microwave network, rather than being carried out in baseband.
Consequently, mutual coupling does not merely perturb the channel, but directly interacts with the analog processing of the \gls{milac}, influencing the implemented linear operator and the resulting end-to-end system behavior.

This observation motivates the need for a rigorous and physically grounded modeling for \gls{milac}-aided systems that explicitly accounts for mutual coupling effects.
Understanding how coupling modifies the effective \gls{milac} operation and impacts beamforming is essential for assessing performance limits and guiding practical design.
In this paper, we fill this gap by developing a comprehensive model of \gls{milac}-aided \gls{mimo} systems in the presence of antenna mutual coupling, and proposing a mutual coupling aware optimization algorithm for \gls{milac}.
Specifically, the contributions of this paper are as follows.

\textit{First}, we develop a rigorous physics-compliant model for \gls{mimo} systems aided by \gls{milac} through multiport network theory, explicitly accounting for antenna mutual coupling.
We derive end-to-end system models for three cases where:
\textit{i)} only the transmitter is equipped with a \gls{milac} (Section~\ref{sec:tx}),
\textit{ii)} only the receiver is equipped with a \gls{milac} (Section~\ref{sec:rx}), and
\textit{iii)} both the transmitter and the receiver are equipped with a \gls{milac} (Section~\ref{sec:both}).
For each of these cases, we show how mutual coupling alters the impact of \gls{milac} on the wireless channel.

\textit{Second}, we formulate and solve an optimization problem for \gls{milac} in the presence of antenna mutual coupling (Section~\ref{sec:opt}).
Considering a \gls{miso} system with \gls{milac} at the transmitter, we derive a global optimal solution that maximizes the received signal power.
Furthermore, we obtain closed-form expressions for the maximum achievable performance and for its average in the presence of uncorrelated fading channels.
Analytical results show that \gls{milac} with mutual coupling achieves better average performance than with no mutual coupling.

\textit{Third}, we compare \gls{milac}-aided beamforming with digital beamforming by examining two possible digital architectures: one equipped with a matching network (or decoupling network) and one without it (Section~\ref{sec:comparison}).
Comparing the optimal performance of \gls{milac}-aided beamforming and digital beamforming, we derive the following two analytical results:
\textit{i)} \gls{milac} achieves the same performance as digital beamforming with a matching network, for any channel realization.
\textit{ii)} in the presence of mutual coupling, \gls{milac} achieves better performance than digital beamforming without a matching network, for any channel realization.

\textit{Fourth}, we validate the presented modeling and optimization algorithm through numerical simulations (Section~\ref{sec:results}).
Numerical results confirm that the proposed \gls{milac} optimization algorithm achieves the derived performance upper bounds, demonstrating its global optimality.
We study the impact of mutual coupling on \gls{milac}-aided systems and conventional digital systems, numerically validating all our theoretical insights, namely that mutual coupling is beneficial in \gls{milac}-aided systems and that \gls{milac} outperforms digital beamforming with no matching network in the presence of mutual coupling.

\textit{Notation}:
Vectors and matrices are denoted with bold lower and bold upper letters, respectively.
Scalars are represented with letters not in bold font.
$\Re\{a\}$, $\Im\{a\}$, and $\vert a\vert$ refer to the real part, imaginary part, and absolute value of a complex scalar $a$, respectively.
$\mathbf{a}^T$, $\mathbf{a}^H$, $[\mathbf{a}]_{i}$, and $\Vert\mathbf{a}\Vert$ refer to the transpose, conjugate transpose, $i$th element, and $l_{2}$-norm of a vector $\mathbf{a}$, respectively.
$\mathbf{A}^T$, $\mathbf{A}^H$, $[\mathbf{A}]_{i,k}$, $[\mathbf{A}]_{i,:}$, and $[\mathbf{A}]_{:,k}$ refer to the transpose, conjugate transpose, $(i,k)$th element, $i$th row, and $k$th column of a matrix $\mathbf{A}$, respectively.
$[\mathbf{A}]_{\mathcal{I},\mathcal{K}}$ refers to the submatrix of $\mathbf{A}$ obtained by selecting the rows and columns indexed by the elements of the sets $\mathcal{I}$ and $\mathcal{K}$, respectively.
$\mathbb{R}$ and $\mathbb{C}$ denote the real and complex number sets, respectively.
$j=\sqrt{-1}$ denotes the imaginary unit.
$\mathbf{I}$ and $\mathbf{0}$ denote the identity matrix and the all-zero matrix with appropriate dimensions.
$\mathbf{A}\succcurlyeq\mathbf{B}$ means that $\mathbf{A}-\mathbf{B}$ is positive semi-definite.

\section{Modeling a MIMO System}
\label{sec:mimo}

To introduce the necessary theoretical tools and the quantities in the considered wireless systems, we first model a conventional digital \gls{mimo} system, where each antenna is directly connected to an \gls{rf} chain.
The transmitting \gls{rf} chains are modeled through voltage generators with their series impedance $Z_0$, e.g., $Z_0=50~\Omega$, while the receiving \gls{rf} chains are modeled as loads $Z_0$, assuming that their input impedances are perfectly matched.
Such a \gls{mimo} system is represented in Fig.~\ref{fig:mimo}, where we have $N_T$ transmitting antennas and $N_R$ receiving antennas.

We denote the voltages at the voltage generators as $\mathbf{s}\in\mathbb{C}^{N_T\times1}$, the voltages and currents at the transmitting antennas as $\mathbf{x}\in\mathbb{C}^{N_T\times1}$ and $\mathbf{i}_x\in\mathbb{C}^{N_T\times1}$, and the voltages and currents at the receiver as $\mathbf{z}\in\mathbb{C}^{N_R\times1}$ and $\mathbf{i}_z\in\mathbb{C}^{N_R\times1}$, as in Fig.~\ref{fig:mimo}.
The received signal $\mathbf{z}$ is related to the transmitted signal $\mathbf{s}$ by
\begin{equation}
\mathbf{z}=\mathbf{H}\mathbf{s},
\end{equation}
where $\mathbf{H}\in\mathbb{C}^{N_R\times N_T}$ is the wireless channel matrix between the transmitter and receiver.
In the remainder of this section, our goal is to characterize $\mathbf{H}$ through rigorous multiport network analysis, recalling what was done in classical \gls{mimo} literature \cite{cle07,ivr10,ivr14,wal04-1,wal04-2}.

\begin{figure}[t]
\centering
\includegraphics[width=0.48\textwidth]{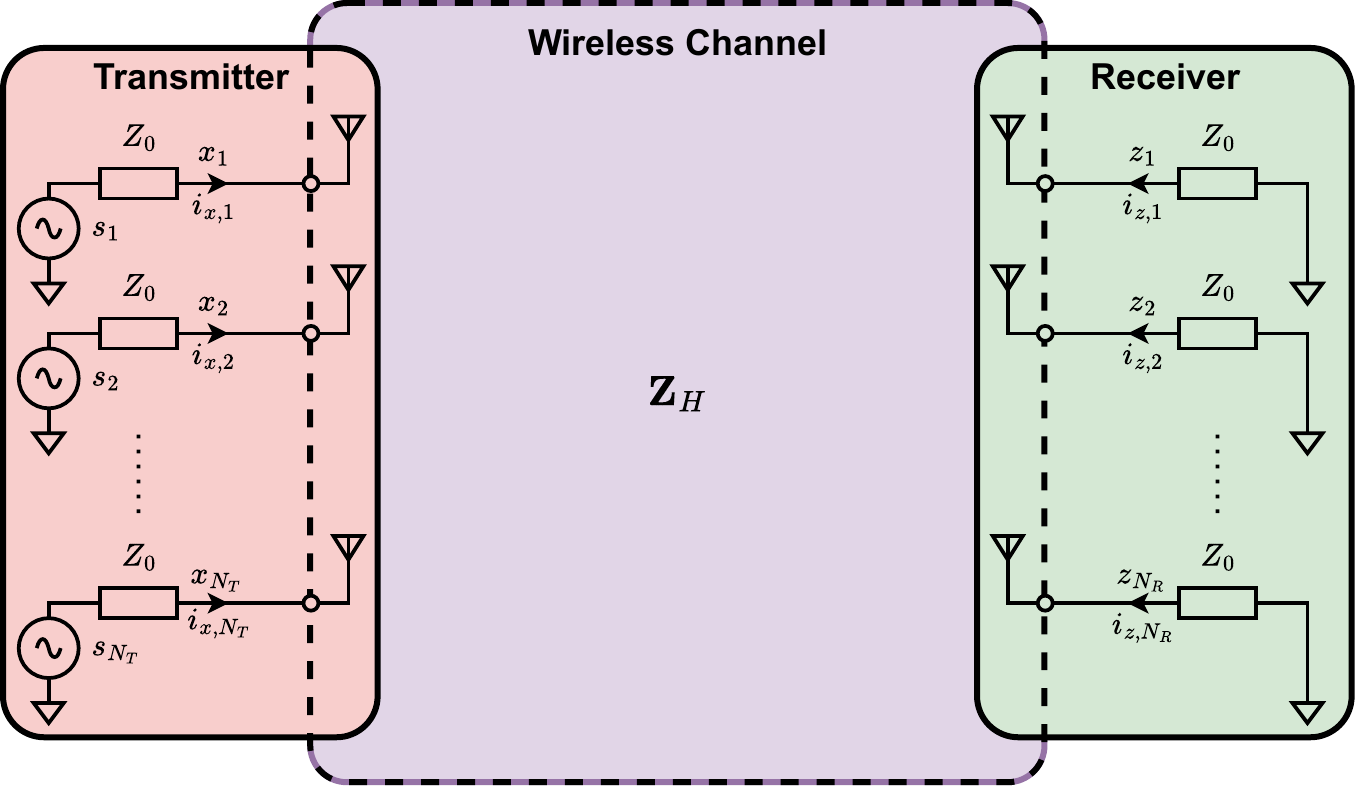}
\caption{Model of a MIMO system.}
\label{fig:mimo}
\end{figure}

\subsection{Model With Mutual Coupling}

To derive $\mathbf{H}$, we first study the relationships between the electrical quantities at the wireless channel, the transmitter, and the receiver, and then solve a system of linear equations involving these relationships.
We model the wireless channel as a giant $(N_T+N_R)$-port network, which can be described by its impedance matrix $\mathbf{Z}_H\in\mathbb{C}^{(N_T+N_R)\times(N_T+N_R)}$ \cite[Chapter~4]{poz11}.
This matrix can be conveniently partitioned as 
\begin{equation}
\mathbf{Z}_H=
\begin{bmatrix}
\mathbf{Z}_{TT} & \mathbf{Z}_{TR}\\
\mathbf{Z}_{RT} & \mathbf{Z}_{RR}
\end{bmatrix},\label{eq:ZH}
\end{equation}
where $\mathbf{Z}_{TT}\in\mathbb{C}^{N_T\times N_T}$ and $\mathbf{Z}_{RR}\in\mathbb{C}^{N_R\times N_R}$ are the impedance matrices at the transmitting and receiving antenna arrays.
The off-diagonal entries of $\mathbf{Z}_{TT}$ and $\mathbf{Z}_{RR}$ are the antenna mutual coupling effects, and the diagonal entries of $\mathbf{Z}_{TT}$ and $\mathbf{Z}_{RR}$ are the antenna self-impedances.
Besides, $\mathbf{Z}_{RT}$ is the transmission impedance matrix from the transmitter to the receiver, and $\mathbf{Z}_{TR}=\mathbf{Z}_{RT}^T$ for the reciprocity of the channel.
Following the definition of impedance matrix \cite[Chapter~4]{poz11}, we have
\begin{equation}
\begin{bmatrix}
\mathbf{x}\\
\mathbf{z}
\end{bmatrix}=
\begin{bmatrix}
\mathbf{Z}_{TT} & \mathbf{Z}_{TR}\\
\mathbf{Z}_{RT} & \mathbf{Z}_{RR}
\end{bmatrix}
\begin{bmatrix}
\mathbf{i}_x\\
\mathbf{i}_z
\end{bmatrix},\label{eq:vZHi-mimo}
\end{equation}
where the directions of the currents are shown in Fig.~\ref{fig:mimo}.
At the transmitter, $\mathbf{x}$ and $\mathbf{i}_x$ are related to the source vector $\mathbf{s}$ by
\begin{equation}
\mathbf{x}=\mathbf{s}-Z_0\mathbf{i}_x,\label{eq:s-mimo}
\end{equation}
which is obtained by applying Ohm's law at all the transmitting antennas, and where the minus sign is in agreement with the current directions in Fig.~\ref{fig:mimo}.
At the receiver, $\mathbf{z}$ and $\mathbf{i}_z$ are related by
\begin{equation}
\mathbf{z}=-Z_0\mathbf{i}_z,\label{eq:y-mimo}
\end{equation}
following Ohm's law.

The channel $\mathbf{H}$ can now be derived by solving the system of the three equations \eqref{eq:vZHi-mimo}, \eqref{eq:s-mimo}, and \eqref{eq:y-mimo}, which is compactly written as
\begin{equation}
\begin{cases}
\mathbf{v}=\mathbf{Z}_H\mathbf{i},\\
\mathbf{v}=\bar{\mathbf{s}}-Z_0\mathbf{i},
\end{cases}\label{eq:sys-mimo-1}
\end{equation}
where we introduced $\mathbf{v}\in\mathbb{C}^{(N_T+N_R)\times1}$, $\mathbf{i}\in\mathbb{C}^{(N_T+N_R)\times1}$, and $\bar{\mathbf{s}}\in\mathbb{C}^{(N_T+N_R)\times1}$ as
\begin{gather}
\mathbf{v}=
\begin{bmatrix}
\mathbf{x}\\
\mathbf{z}
\end{bmatrix},\;
\mathbf{i}=
\begin{bmatrix}
\mathbf{i}_x\\
\mathbf{i}_z
\end{bmatrix},\;
\bar{\mathbf{s}}=
\begin{bmatrix}
\mathbf{s}\\
\mathbf{0}
\end{bmatrix}.
\end{gather}
System \eqref{eq:sys-mimo-1} gives $\mathbf{v}=\mathbf{Z}_H\left(\mathbf{Z}_H+Z_0\mathbf{I}\right)^{-1}\bar{\mathbf{s}}$.
Thus, by introducing $\mathbf{A}\in\mathbb{C}^{(N_T+N_R)\times(N_T+N_R)}$ as
\begin{equation}
\mathbf{A}=\mathbf{Z}_H\left(\mathbf{Z}_H+Z_0\mathbf{I}\right)^{-1},\label{eq:A-mimo-1}
\end{equation}
partitioned as
\begin{equation}
\mathbf{A}=
\begin{bmatrix}
\mathbf{A}_{11} & \mathbf{A}_{12}\\
\mathbf{A}_{21} & \mathbf{A}_{22}
\end{bmatrix},\label{eq:A-mimo-partition}
\end{equation}
with $\mathbf{A}_{11}\in\mathbb{C}^{N_T\times N_T}$, $\mathbf{A}_{22}\in\mathbb{C}^{N_R\times N_R}$, we have $\mathbf{z}=\mathbf{A}_{21}\mathbf{s}$.
By comparing the relationship $\mathbf{z}=\mathbf{A}_{21}\mathbf{s}$ with $\mathbf{z}=\mathbf{H}\mathbf{s}$, we notice that $\mathbf{H}=\mathbf{A}_{21}$, which is derived in the following.

To simplify the derivation, we consider the unilateral approximation \cite{ivr10}, i.e., we assume that the electrical properties at the transmitter are independent of the electrical properties at the receiver.
This assumption accurately reflects what happens in common wireless systems and allows us to neglect the feedback channel by setting it as $\mathbf{Z}_{TR}=\mathbf{0}$.
Expanding $\mathbf{Z}_H$ in \eqref{eq:A-mimo-1}, we can therefore write $\mathbf{A}$ as
\begin{equation}
\mathbf{A}=
\begin{bmatrix}
\mathbf{Z}_{TT} & \mathbf{0}\\
\mathbf{Z}_{RT} & \mathbf{Z}_{RR}
\end{bmatrix}
\begin{bmatrix}
\mathbf{Z}_{TT}+Z_0\mathbf{I} & \mathbf{0}\\
\mathbf{Z}_{RT} & \mathbf{Z}_{RR}+Z_0\mathbf{I}
\end{bmatrix}^{-1}.\label{eq:A-mimo-2}
\end{equation}
By carrying out the matrix inverse in \eqref{eq:A-mimo-2} (note that the matrix to be inverted is a lower triangular $2\times 2$ block matrix) and the matrix product, we obtain
\begin{equation}
\mathbf{H}
=\mathbf{A}_{21}
=Z_0\left(\mathbf{Z}_{RR}+Z_0\mathbf{I}\right)^{-1}\mathbf{Z}_{RT}\left(\mathbf{Z}_{TT}+Z_0\mathbf{I}\right)^{-1},\label{eq:H-mimo}
\end{equation}
giving our desired wireless channel matrix, in agreement with \cite[Section~II]{cle07}.

\subsection{Model With No Mutual Coupling}

The obtained channel model can be simplified by assuming that all antennas are perfectly matched to $Z_0$ and neglecting the mutual coupling effects, which is realistic when the antenna spacing is at least half-wavelength.
First, when all antennas at the transmitter and receiver are perfectly matched to $Z_0$, namely their self-impedances are all $Z_0$, the diagonal entries of $\mathbf{Z}_{TT}$ and $\mathbf{Z}_{RR}$ are all equal to $Z_0$.
Second, with no mutual coupling, the off-diagonal entries of $\mathbf{Z}_{TT}$ and $\mathbf{Z}_{RR}$ are all zero.
Therefore, under these two assumptions, we have $\mathbf{Z}_{TT}=Z_0\mathbf{I}$ and $\mathbf{Z}_{RR}=Z_0\mathbf{I}$, which simplify the channel model in \eqref{eq:H-mimo} as $\mathbf{H}=\mathbf{Z}_{RT}/(4Z_0)$,
purely depending on the transmission impedance matrix $\mathbf{Z}_{RT}$.

\section{Modeling a MIMO System with MiLAC\\at the Transmitter}
\label{sec:tx}

\begin{figure}[t]
\centering
\includegraphics[width=0.48\textwidth]{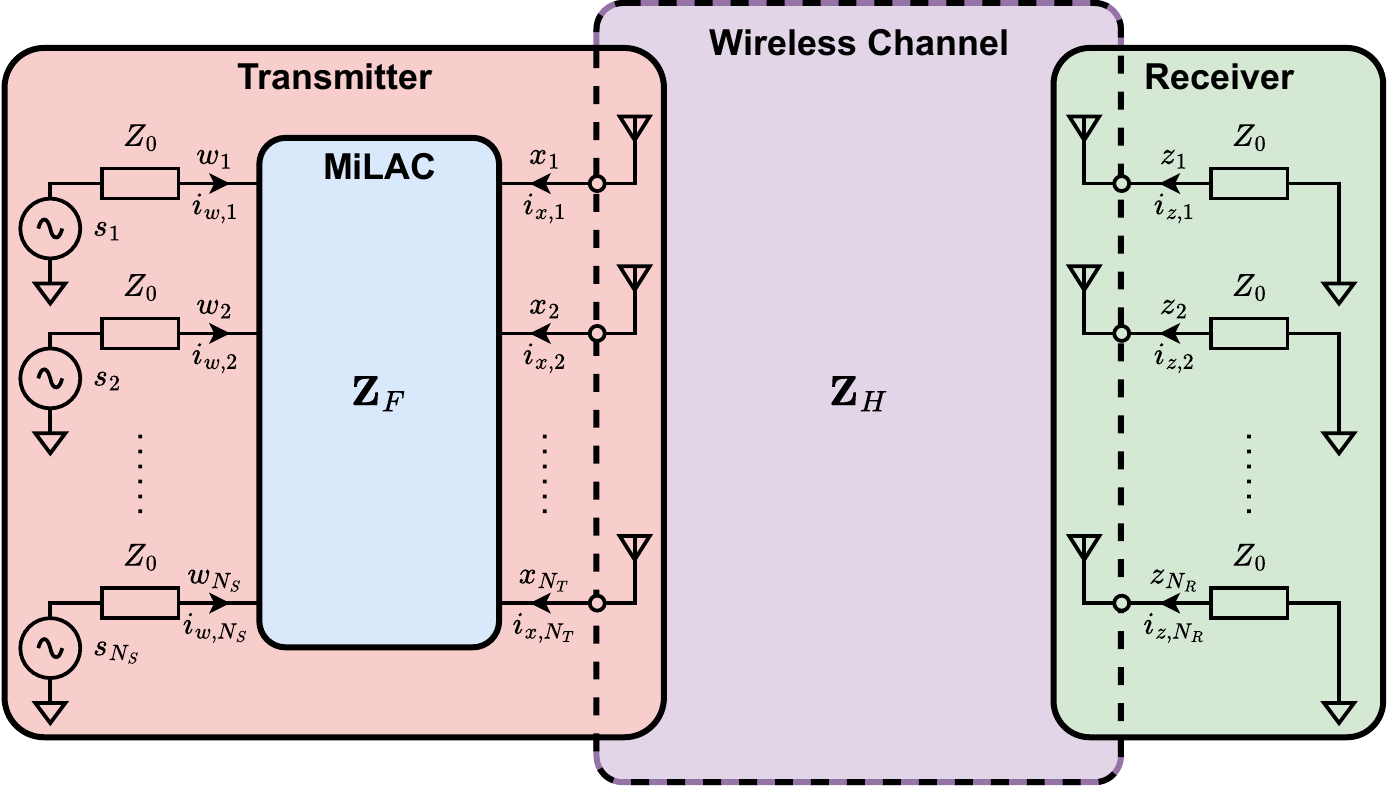}
\caption{Model of a MIMO system with MiLAC at the transmitter.}
\label{fig:tx}
\end{figure}

We have analyzed a conventional digital \gls{mimo} system with multiport network analysis.
In this section, we show how to apply the same tools to analyze a \gls{mimo} system where the transmitter is equipped with a \gls{milac}, as shown in Fig.~\ref{fig:tx}.

At the transmitter, we denote the voltages at the generators as $\mathbf{s}\in\mathbb{C}^{N_S\times1}$, where $N_S$ is the number of \gls{rf} chains, the voltages and currents at the input of the \gls{milac} as $\mathbf{w}\in\mathbb{C}^{N_S\times1}$ and $\mathbf{i}_w\in\mathbb{C}^{N_S\times1}$, and the voltages and currents at the transmitting antennas as $\mathbf{x}\in\mathbb{C}^{N_T\times1}$ and $\mathbf{i}_x\in\mathbb{C}^{N_T\times1}$, where $N_T$ is the number of transmitting antennas.
At the receiver, we denote the voltages and currents as $\mathbf{z}\in\mathbb{C}^{N_R\times1}$ and $\mathbf{i}_z\in\mathbb{C}^{N_R\times1}$, where $N_R$ is the number of receiving antennas (which is the same as the number of \gls{rf} chains at the receiver since it operates digital beamforming).
The received signal $\mathbf{z}$ is given as a function of the transmitted signal $\mathbf{s}$ by
\begin{equation}
\mathbf{z}=\mathbf{H}\mathbf{F}\mathbf{s},
\end{equation}
where $\mathbf{H}\in\mathbb{C}^{N_R\times N_T}$ is the wireless channel matrix and $\mathbf{F}\in\mathbb{C}^{N_T\times N_S}$ is the precoding matrix implemented by the \gls{milac}.
In the following, we derive accurate expressions for $\mathbf{H}$ and $\mathbf{F}$ in the presence of mutual coupling at both the transmitter and receiver.

\subsection{Model With Mutual Coupling}

As it was done for the digital \gls{mimo} system in Section~\ref{sec:mimo}, we first derive the relationships between the signals at the wireless channel, the transmitter, and the receiver, and then consider them jointly to derive the system model.
The wireless channel can be seen as a $(N_T+N_R)$-port network with impedance matrix $\mathbf{Z}_H\in\mathbb{C}^{(N_T+N_R)\times(N_T+N_R)}$, partitioned as in \eqref{eq:ZH}.
Therefore, according to the definition of impedance matrix \cite[Chapter~4]{poz11}, the voltages and currents at the transmitting antennas and at the receiver are related by
\begin{equation}
\begin{bmatrix}
\mathbf{x}\\
\mathbf{z}
\end{bmatrix}=
\begin{bmatrix}
\mathbf{Z}_{TT} & \mathbf{Z}_{TR}\\
\mathbf{Z}_{RT} & \mathbf{Z}_{RR}
\end{bmatrix}
\begin{bmatrix}
-\mathbf{i}_x\\
\mathbf{i}_z
\end{bmatrix},\label{eq:vZHi-tx}
\end{equation}
where the directions of the currents $\mathbf{i}_x$ and $\mathbf{i}_z$ are illustrated in Fig.~\ref{fig:tx}.
At the transmitter, $\mathbf{w}$ and $\mathbf{i}_w$ are related to the source vector $\mathbf{s}$ by
\begin{equation}
\mathbf{w}=\mathbf{s}-Z_0\mathbf{i}_w,\label{eq:s-tx}
\end{equation}
following Ohm's law.
In addition, introducing the impedance matrix of the \gls{milac} as $\mathbf{Z}_F\in\mathbb{C}^{(N_S+N_T)\times(N_S+N_T)}$, partitioned as
\begin{equation}
\mathbf{Z}_F=
\begin{bmatrix}
\mathbf{Z}_{F,11} & \mathbf{Z}_{F,12}\\
\mathbf{Z}_{F,21} & \mathbf{Z}_{F,22}
\end{bmatrix},\label{eq:ZF}
\end{equation}
where $\mathbf{Z}_{F,11}\in\mathbb{C}^{N_S\times N_S}$ and $\mathbf{Z}_{F,22}\in\mathbb{C}^{N_T\times N_T}$, we have
\begin{equation}
\begin{bmatrix}
\mathbf{w}\\
\mathbf{x}
\end{bmatrix}
=
\begin{bmatrix}
\mathbf{Z}_{F,11} & \mathbf{Z}_{F,12}\\
\mathbf{Z}_{F,21} & \mathbf{Z}_{F,22}
\end{bmatrix}
\begin{bmatrix}
\mathbf{i}_w\\
\mathbf{i}_x
\end{bmatrix},\label{eq:vZFi}
\end{equation}
by the definition of impedance matrix \cite[Chapter~4]{poz11}.
Finally, at the receiver, we have
\begin{equation}
\mathbf{z}=-Z_0\mathbf{i}_z,\label{eq:y-tx}
\end{equation}
as for the conventional \gls{mimo} system in Section~\ref{sec:mimo}.

To derive the channel $\mathbf{H}$ and the precoder $\mathbf{F}$, we need now to solve the system of the four equations \eqref{eq:vZHi-tx}, \eqref{eq:s-tx}, \eqref{eq:vZFi}, and \eqref{eq:y-tx}.
Considering the unilateral approximation \cite{ivr10}, i.e., $\mathbf{Z}_{TR}=\mathbf{0}$, this system can be compactly written as
\begin{equation}
\begin{cases}
\mathbf{v}=\mathbf{Z}\mathbf{i},\\
\mathbf{v}=\bar{\mathbf{s}}-\overline{\mathbf{Z}}\mathbf{i},
\end{cases}\label{eq:sys-tx-1}
\end{equation}
where we introduced $\mathbf{v}\in\mathbb{C}^{(N_S+N_T+N_R)\times1}$, $\mathbf{i}\in\mathbb{C}^{(N_S+N_T+N_R)\times1}$, $\mathbf{Z}\in\mathbb{C}^{(N_S+N_T+N_R)\times(N_S+N_T+N_R)}$, $\bar{\mathbf{s}}\in\mathbb{C}^{(N_S+N_T+N_R)\times1}$, and $\overline{\mathbf{Z}}\in\mathbb{C}^{(N_S+N_T+N_R)\times(N_S+N_T+N_R)}$ as
\begin{gather}
\mathbf{v}=
\begin{bmatrix}
\mathbf{w}\\
\mathbf{x}\\
\mathbf{z}
\end{bmatrix},
\mathbf{i}=
\begin{bmatrix}
\mathbf{i}_w\\
\mathbf{i}_x\\
\mathbf{i}_z
\end{bmatrix},
\mathbf{Z}=
\begin{bmatrix}
\mathbf{Z}_{F,11} & \mathbf{Z}_{F,12} & \mathbf{0}\\
\mathbf{Z}_{F,21} & \mathbf{Z}_{F,22} & \mathbf{0}\\
\mathbf{0} & -\mathbf{Z}_{RT} & \mathbf{Z}_{RR}
\end{bmatrix},\\
\bar{\mathbf{s}}=
\begin{bmatrix}
\mathbf{s}\\
\mathbf{0}\\
\mathbf{0}
\end{bmatrix},\;
\overline{\mathbf{Z}}=
\begin{bmatrix}
Z_0\mathbf{I} & \mathbf{0} & \mathbf{0}\\
\mathbf{0} & \mathbf{Z}_{TT} & \mathbf{0}\\
\mathbf{0} & \mathbf{0} & Z_0\mathbf{I}
\end{bmatrix}.
\end{gather}
System \eqref{eq:sys-tx-1} gives $\mathbf{v}=\mathbf{Z}\left(\mathbf{Z}+\overline{\mathbf{Z}}\right)^{-1}\bar{\mathbf{s}}$.
Thus, by introducing $\mathbf{A}\in\mathbb{C}^{(N_S+N_T+N_R)\times(N_S+N_T+N_R)}$ as
\begin{equation}
\mathbf{A}=\mathbf{Z}\left(\mathbf{Z}+\overline{\mathbf{Z}}\right)^{-1},
\end{equation}
partitioned as
\begin{equation}
\mathbf{A}=
\begin{bmatrix}
\mathbf{A}_{11} & \mathbf{A}_{12} & \mathbf{A}_{13}\\
\mathbf{A}_{21} & \mathbf{A}_{22} & \mathbf{A}_{23}\\
\mathbf{A}_{31} & \mathbf{A}_{32} & \mathbf{A}_{33}
\end{bmatrix},\label{eq:A-partition}
\end{equation}
with $\mathbf{A}_{11}\in\mathbb{C}^{N_S\times N_S}$, $\mathbf{A}_{22}\in\mathbb{C}^{N_T\times N_T}$, and $\mathbf{A}_{33}\in\mathbb{C}^{N_R\times N_R}$, we have $\mathbf{z}=\mathbf{A}_{31}\mathbf{s}$.
By comparing the relationship $\mathbf{z}=\mathbf{A}_{31}\mathbf{s}$ with $\mathbf{z}=\mathbf{H}\mathbf{F}\mathbf{s}$, we notice that $\mathbf{H}\mathbf{F}=\mathbf{A}_{31}$.

To obtain the expression of $\mathbf{A}_{31}$, we introduce $\widetilde{\mathbf{Z}}_{RT}\in\mathbb{C}^{N_R\times(N_S+N_T)}$ and $\widetilde{\mathbf{Z}}_{TT}\in\mathbb{C}^{(N_S+N_T)\times(N_S+N_T)}$ as
\begin{equation}
\widetilde{\mathbf{Z}}_{RT}=
\begin{bmatrix}
\mathbf{0} & \mathbf{Z}_{RT}
\end{bmatrix},\;
\widetilde{\mathbf{Z}}_{TT}=
\begin{bmatrix}
Z_0\mathbf{I} & \mathbf{0}\\
\mathbf{0} & \mathbf{Z}_{TT}
\end{bmatrix},\label{eq:tilde-tx}
\end{equation}
such that $\mathbf{A}=\mathbf{Z}\left(\mathbf{Z}+\overline{\mathbf{Z}}\right)^{-1}$ can be rewritten as
\begin{equation}
\mathbf{A}=
\begin{bmatrix}
\mathbf{Z}_F & \mathbf{0}\\
-\widetilde{\mathbf{Z}}_{RT} & \mathbf{Z}_{RR}
\end{bmatrix}
\begin{bmatrix}
\mathbf{Z}_F+\widetilde{\mathbf{Z}}_{TT} & \mathbf{0}\\
-\widetilde{\mathbf{Z}}_{RT} & \mathbf{Z}_{RR}+Z_0\mathbf{I}
\end{bmatrix}^{-1}.\label{eq:A-tx}
\end{equation}
By carrying out the matrix inverse in \eqref{eq:A-tx} (note that the matrix to be inverted is a lower triangular $2\times 2$ block matrix) and the matrix product, we obtain
\begin{multline}
\begin{bmatrix}
\mathbf{A}_{31} & \mathbf{A}_{32}
\end{bmatrix}
=\\
-Z_0\left(\mathbf{Z}_{RR}+Z_0\mathbf{I}\right)^{-1}\widetilde{\mathbf{Z}}_{RT}\left(\mathbf{Z}_F+\widetilde{\mathbf{Z}}_{TT}\right)^{-1}.\label{eq:A31-32}
\end{multline}

By substituting \eqref{eq:tilde-tx} into \eqref{eq:A31-32}, and carrying out the necessary computations, we obtain
\begin{multline}
\mathbf{A}_{31}
=Z_0\left(\mathbf{Z}_{RR}+Z_0\mathbf{I}\right)^{-1}\mathbf{Z}_{RT}\mathbf{Z}_{TT}^{-1}\\
\times\left[\left(\frac{\mathbf{Y}_F}{Y_0}+
\begin{bmatrix}
\mathbf{I} & \mathbf{0}\\
\mathbf{0} & \frac{\mathbf{Z}_{TT}^{-1}}{Y_0}
\end{bmatrix}
\right)^{-1}
\right]_{N_S+(1:N_T),1:N_S},
\end{multline}
where we introduced $\mathbf{Y}_F=\mathbf{Z}_F^{-1}$ as the admittance matrix of the \gls{milac} and $Y_0=Z_0^{-1}$.
Since $\mathbf{H}\mathbf{F}=\mathbf{A}_{31}$, we identify the channel matrix $\mathbf{H}$ and the \gls{milac} precoding matrix $\mathbf{F}$ as
\begin{equation}
\mathbf{H}
=Z_0\left(\mathbf{Z}_{RR}+Z_0\mathbf{I}\right)^{-1}\mathbf{Z}_{RT}\mathbf{Z}_{TT}^{-1},\label{eq:H-tx}
\end{equation}
and
\begin{equation}
\mathbf{F}
=\left[\left(\frac{\mathbf{Y}_F}{Y_0}+
\begin{bmatrix}
\mathbf{I} & \mathbf{0}\\
\mathbf{0} & \frac{\mathbf{Z}_{TT}^{-1}}{Y_0}
\end{bmatrix}\right)^{-1}\right]_{N_S+(1:N_T),1:N_S},\label{eq:F-tx}
\end{equation}
showing that $\mathbf{F}$ depends not only on the admittance matrix of the \gls{milac} $\mathbf{Y}_F$, but also on the mutual coupling matrix at the transmitter $\mathbf{Z}_{TT}$.

Note that by plugging \eqref{eq:tilde-tx} into \eqref{eq:A31-32}, we can also obtain an alternative expression of the end-to-end channel $\mathbf{A}_{31}$ following different steps, namely
\begin{equation}
\mathbf{A}_{31}=Z_0\left(\mathbf{Z}_{RR}+Z_0\mathbf{I}\right)^{-1}\mathbf{Z}_{RT}\mathbf{J}_{T}^T\left(\mathbf{Z}_{T}+Z_0\mathbf{I}\right)^{-1},
\end{equation}
where we introduced $\mathbf{Z}_{T}\in\mathbb{C}^{N_S\times N_S}$ and $\mathbf{J}_{T}\in\mathbb{C}^{N_S\times N_T}$ as
\begin{gather}
\mathbf{Z}_{T}=\mathbf{Z}_{F,11}-\mathbf{J}_{T}\mathbf{Z}_{F,21},\label{eq:ZT}\\
\mathbf{J}_{T}=\mathbf{Z}_{F,12}\left(\mathbf{Z}_{F,22}+\mathbf{Z}_{TT}\right)^{-1}.\label{eq:JT}
\end{gather}
This expression, based on the \gls{milac} impedance matrix $\mathbf{Z}_{F}$ rather than the \gls{milac} admittance matrix $\mathbf{Y}_{F}$, has been used to study the impact of matching networks on \gls{mimo} channels \cite[Section~II]{ivr10}.
Therefore, a \gls{milac} has the same impact on the channel expression as a matching network.
Nevertheless, while a matching network is fixed and it is designed to maximize the power flow from the generators to the antennas in the presence of mutual coupling, a \gls{milac} is reconfigurable and can perform beamforming depending on the channel realization.
A \gls{milac} alleviates the baseband processing needs at the cost of a reconfigurable microwave network.
In terms of representation, it is convenient to represent the effect of a \gls{milac} with its admittance matrix as in \eqref{eq:F-tx}, since it is closely related to its tunable components, as specified in \cite{ner25-1}.

\subsection{Model With No Mutual Coupling}

Assuming perfect matching and no mutual coupling at both the transmitter and receiver, the expressions of the channel matrix $\mathbf{H}$ and precoding matrix $\mathbf{F}$ can be significantly simplified.
In particular, as a consequence of perfect matching and with negligible mutual coupling, we have $\mathbf{Z}_T=Z_0\mathbf{I}$ and $\mathbf{Z}_R=Z_0\mathbf{I}$.
By substituting $\mathbf{Z}_T=Z_0\mathbf{I}$ and $\mathbf{Z}_R=Z_0\mathbf{I}$ into \eqref{eq:H-tx} and \eqref{eq:F-tx}, we obtain $\mathbf{H}=\mathbf{Z}_{RT}/(2Z_0)$
and
\begin{equation}
\mathbf{F}=\left[\left(\frac{\mathbf{Y}_F}{Y_0}+\mathbf{I}\right)^{-1}\right]_{N_S+(1:N_T),1:N_S},\label{eq:F-noMC-tx}
\end{equation}
in agreement with previous literature on \gls{milac} \cite{ner25-1,ner25-2}.

\section{Modeling a MIMO System with MiLAC\\at the Receiver}
\label{sec:rx}

\begin{figure}[t]
\centering
\includegraphics[width=0.48\textwidth]{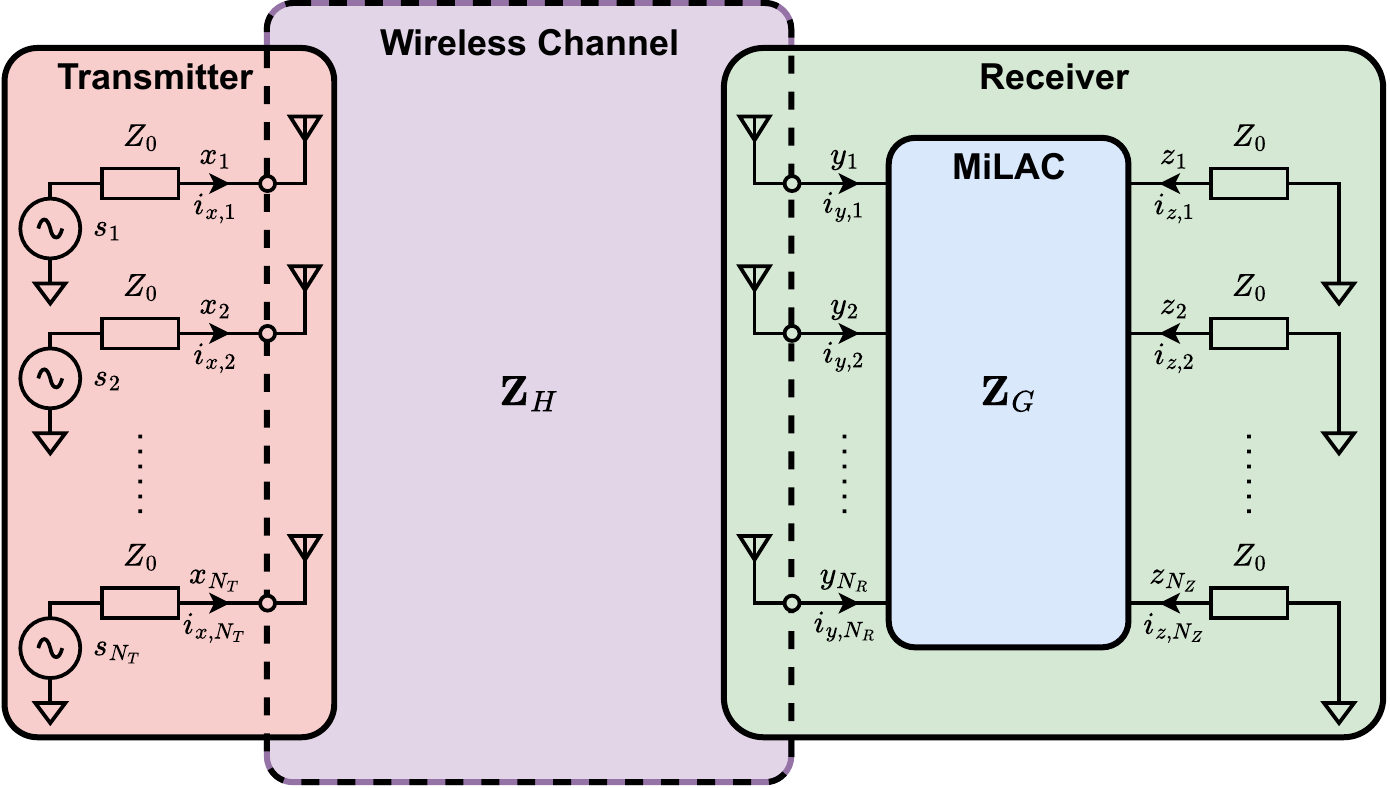}
\caption{Model of a MIMO system with MiLAC at the receiver.}
\label{fig:rx}
\end{figure}

We have analyzed a \gls{mimo} system with \gls{milac} at the transmitter, and derived the expressions of the channel and precoding matrix with and without mutual coupling effects.
In this section, we carry out the same analysis for a \gls{mimo} system where the receiver is equipped with a \gls{milac}, as shown in Fig.~\ref{fig:rx}.

At the transmitter, we denote the voltages at the voltage generators as $\mathbf{s}\in\mathbb{C}^{N_T\times1}$, and the voltages and currents at the transmitting antennas as $\mathbf{x}\in\mathbb{C}^{N_T\times1}$ and $\mathbf{i}_x\in\mathbb{C}^{N_T\times1}$, where $N_T$ is the number of antennas.
At the receiver, we denote the voltages and currents at the receiving antennas as $\mathbf{y}\in\mathbb{C}^{N_R\times1}$ and $\mathbf{i}_y\in\mathbb{C}^{N_R\times1}$, where $N_R$ is the number of antennas, and the voltages and currents at the receiving \gls{rf} chains as $\mathbf{z}\in\mathbb{C}^{N_Z\times1}$ and $\mathbf{i}_z\in\mathbb{C}^{N_Z\times1}$, where $N_Z$ the number of receiving \gls{rf} chains.
The received signal $\mathbf{z}$ is given as a function of the transmitted signal $\mathbf{s}$ by
\begin{equation}
\mathbf{z}=\mathbf{G}\mathbf{H}\mathbf{s},
\end{equation}
where $\mathbf{G}\in\mathbb{C}^{N_Z\times N_R}$ is the combining matrix implemented by the \gls{milac} and $\mathbf{H}\in\mathbb{C}^{N_R\times N_T}$ is the wireless channel matrix.
Our goal is to derive accurate expressions for $\mathbf{G}$ and $\mathbf{H}$ in the presence of mutual coupling at both the transmitter and receiver.

\subsection{Model With Mutual Coupling}

We begin by deriving the relationships between the signals at the wireless channel, the transmitter, and the receiver, and then solve a system of equations involving all these relationships.
We represent the wireless channel as a $(N_T+N_R)$-port network with impedance matrix $\mathbf{Z}_H\in\mathbb{C}^{(N_T+N_R)\times(N_T+N_R)}$, partitioned as in \eqref{eq:ZH}.
Following the definition of impedance matrix, the voltages and currents at the transmitting and receiving antennas are related by
\begin{equation}
\begin{bmatrix}
\mathbf{x}\\
\mathbf{y}
\end{bmatrix}=
\begin{bmatrix}
\mathbf{Z}_{TT} & \mathbf{Z}_{TR}\\
\mathbf{Z}_{RT} & \mathbf{Z}_{RR}
\end{bmatrix}
\begin{bmatrix}
\mathbf{i}_x\\
-\mathbf{i}_y
\end{bmatrix},\label{eq:vZHi-rx}
\end{equation}
where the direction of the currents is shown in Fig.~\ref{fig:rx}.
At the transmitter, we have
\begin{equation}
\mathbf{x}=\mathbf{s}-Z_0\mathbf{i}_x,\label{eq:s-rx}
\end{equation}
as for the conventional \gls{mimo} system in Section~\ref{sec:mimo}.
At the receiver, we introduce the impedance matrix of the \gls{milac} as $\mathbf{Z}_G\in\mathbb{C}^{(N_R+N_Z)\times(N_R+N_Z)}$, partitioned as
\begin{equation}
\mathbf{Z}_G=
\begin{bmatrix}
\mathbf{Z}_{G,11} & \mathbf{Z}_{G,12}\\
\mathbf{Z}_{G,21} & \mathbf{Z}_{G,22}
\end{bmatrix},\label{eq:ZG}
\end{equation}
where $\mathbf{Z}_{G,11}\in\mathbb{C}^{N_R\times N_R}$ and $\mathbf{Z}_{G,22}\in\mathbb{C}^{N_Z\times N_Z}$, which gives
\begin{equation}
\begin{bmatrix}
\mathbf{y}\\
\mathbf{z}
\end{bmatrix}
=
\begin{bmatrix}
\mathbf{Z}_{G,11} & \mathbf{Z}_{G,12}\\
\mathbf{Z}_{G,21} & \mathbf{Z}_{G,22}
\end{bmatrix}
\begin{bmatrix}
\mathbf{i}_y\\
\mathbf{i}_z
\end{bmatrix},\label{eq:vZGi}
\end{equation}
by the definition of impedance matrix \cite[Chapter~4]{poz11}.
In addition, $\mathbf{z}$ and $\mathbf{i}_z$ are related by
\begin{equation}
\mathbf{z}=-Z_0\mathbf{i}_z,\label{eq:z-rx}
\end{equation}
by Ohm's law.

To derive the channel $\mathbf{H}$ and the combiner $\mathbf{G}$, we solve the system of the four equations \eqref{eq:vZHi-rx}, \eqref{eq:s-rx}, \eqref{eq:vZGi}, and \eqref{eq:z-rx}.
With the unilateral approximation \cite{ivr10}, i.e., $\mathbf{Z}_{TR}=\mathbf{0}$, it is possible to show that a procedure similar to the one used in Section~\ref{sec:tx} yields
\begin{multline}
\mathbf{G}\mathbf{H}
=\left[\left(\frac{\mathbf{Y}_G}{Y_0}+
\begin{bmatrix}
\frac{\mathbf{Z}_{RR}^{-1}}{Y_0} & \mathbf{0}\\
\mathbf{0} & \mathbf{I}
\end{bmatrix}
\right)^{-1}\right]_{N_R+(1:N_Z),1:N_R}\\
\times Z_0\mathbf{Z}_{RR}^{-1}\mathbf{Z}_{RT}\left(\mathbf{Z}_{TT}+Z_0\mathbf{I}\right)^{-1},
\end{multline}
where we introduced $\mathbf{Y}_G=\mathbf{Z}_G^{-1}$ as the admittance matrix of the \gls{milac}.
Therefore, we identify the channel $\mathbf{H}$ and the combiner $\mathbf{G}$ as
\begin{equation}
\mathbf{H}
=Z_0\mathbf{Z}_{RR}^{-1}\mathbf{Z}_{RT}\left(\mathbf{Z}_{TT}+Z_0\mathbf{I}\right)^{-1},\label{eq:H-rx}
\end{equation}
and
\begin{equation}
\mathbf{G}=\left[\left(\frac{\mathbf{Y}_G}{Y_0}+
\begin{bmatrix}
\frac{\mathbf{Z}_{RR}^{-1}}{Y_0} & \mathbf{0}\\
\mathbf{0} & \mathbf{I}
\end{bmatrix}
\right)^{-1}\right]_{N_R+(1:N_Z),1:N_R},\label{eq:G-rx}
\end{equation}
showing that the combiner $\mathbf{G}$ depends on the admittance matrix of the \gls{milac} $\mathbf{Y}_G$ and also on the mutual coupling matrix at the receiver $\mathbf{Z}_{RR}$.

It is worthwhile to compare the channel expressions obtained for \gls{milac} at the transmitter and receiver given by \eqref{eq:H-tx} and \eqref{eq:H-rx}, respectively.
To distinguish between them, we denote the channels in \eqref{eq:H-tx} and \eqref{eq:H-rx} as $\mathbf{H}_{\text{Tx}}$ and $\mathbf{H}_{\text{Rx}}$, respectively.
Since $\mathbf{H}_{\text{Tx}}$ and $\mathbf{H}_{\text{Rx}}$ are the channel expressions of the same wireless system where the roles of the transmitter and receiver are swapped, it should hold $\mathbf{H}_{\text{Tx}}^T=\mathbf{H}_{\text{Rx}}$ in the case of a reciprocal channel, which is demonstrated in the following.
By taking the transpose of $\mathbf{H}_{\text{Tx}}$ in \eqref{eq:H-tx}, we obtain $\mathbf{H}_{\text{Tx}}^T=Z_0(\mathbf{Z}_{TT}^{-1})^T\mathbf{Z}_{RT}^T((\mathbf{Z}_{RR}+Z_0\mathbf{I})^{-1})^T$, which can be rewritten as $\mathbf{H}_{\text{Tx}}^T=Z_0\mathbf{Z}_{TT}^{-1}\mathbf{Z}_{TR}(\mathbf{Z}_{RR}+Z_0\mathbf{I})^{-1}$ since $\mathbf{Z}_{TT}=\mathbf{Z}_{TT}^T$, $\mathbf{Z}_{TR}=\mathbf{Z}_{RT}^T$, and $\mathbf{Z}_{RR}=\mathbf{Z}_{RR}^T$ in the case of a reciprocal channel.
By swapping the roles of the transmitter and receiver, i.e., by substituting $\mathbf{Z}_{TT}$, $\mathbf{Z}_{TR}$, and $\mathbf{Z}_{RR}$ with $\mathbf{Z}_{RR}$, $\mathbf{Z}_{RT}$, and $\mathbf{Z}_{TT}$, respectively, we have that the expression of $\mathbf{H}_{\text{Tx}}^T$ coincides with \eqref{eq:H-rx}, showing that $\mathbf{H}_{\text{Tx}}^T=\mathbf{H}_{\text{Rx}}$.

By solving the system of equations \eqref{eq:vZHi-rx}, \eqref{eq:s-rx}, \eqref{eq:vZGi}, and \eqref{eq:z-rx} with alternative derivations, we can also obtain a different but equivalent expression of the end-to-end channel $\mathbf{G}\mathbf{H}$, i.e.,
\begin{equation}
\mathbf{G}\mathbf{H}=Z_0\left(\mathbf{Z}_{R}+Z_0\mathbf{I}\right)^{-1}\mathbf{J}_{R}\mathbf{Z}_{RT}\left(\mathbf{Z}_{TT}+Z_0\mathbf{I}\right)^{-1},
\end{equation}
where $\mathbf{Z}_{R}\in\mathbb{C}^{N_Z\times N_Z}$ and $\mathbf{J}_{R}\in\mathbb{C}^{N_Z\times N_R}$ are
\begin{gather}
\mathbf{Z}_{R}=\mathbf{Z}_{G,22}-\mathbf{J}_{R}\mathbf{Z}_{G,12},\\
\mathbf{J}_{R}=\mathbf{Z}_{G,21}\left(\mathbf{Z}_{G,11}+\mathbf{Z}_{RR}\right)^{-1}.
\end{gather}
This expression depends on the \gls{milac} impedance matrix $\mathbf{Z}_{G}$ rather than the admittance matrix $\mathbf{Y}_{G}$, and has been used to study the impact of matching networks in \gls{mimo} channels \cite[Section~II]{ivr10}.
Therefore, a \gls{milac} at the receiver has the same impact on the channel expression as a matching network.
In terms of representation, it is convenient to represent the effect of a \gls{milac} with its admittance matrix as in \eqref{eq:G-rx}, since it is closely related to its tunable components \cite{ner25-1}.

\subsection{Model With No Mutual Coupling}

We now assume all antennas to be perfectly matched to $Z_0$ and that the mutual coupling effects are negligible, which jointly give $\mathbf{Z}_T=Z_0\mathbf{I}$ and $\mathbf{Z}_R=Z_0\mathbf{I}$.
By substituting $\mathbf{Z}_T=Z_0\mathbf{I}$ and $\mathbf{Z}_R=Z_0\mathbf{I}$ into \eqref{eq:H-rx} and \eqref{eq:G-rx}, we obtain that the channel matrix is given by $\mathbf{H}=\mathbf{Z}_{RT}/(2Z_0)$ as in Section~\ref{sec:tx}, while the combining matrix of the \gls{milac} boils down to
\begin{equation}
\mathbf{G}=\left[\left(\frac{\mathbf{Y}_G}{Y_0}+\mathbf{I}\right)^{-1}\right]_{N_R+(1:N_Z),1:N_R},\label{eq:G-noMC-rx}
\end{equation}
which agrees with prior works on \gls{milac} \cite{ner25-1,ner25-2}.

\section{Modeling a MIMO System with MiLAC\\at both the Transmitter and Receiver}
\label{sec:both}

In this section, we analyze a \gls{mimo} system where both the transmitter and receiver have a \gls{milac}, as shown in Fig.~\ref{fig:both}.

At the transmitter, we denote the source voltages as $\mathbf{s}\in\mathbb{C}^{N_S\times1}$, where $N_S$ is the number of \gls{rf} chains, the voltages and currents at the input of the \gls{milac} as $\mathbf{w}\in\mathbb{C}^{N_S\times1}$ and $\mathbf{i}_w\in\mathbb{C}^{N_S\times1}$, and the voltages and currents at the transmitting antennas as $\mathbf{x}\in\mathbb{C}^{N_T\times1}$ and $\mathbf{i}_x\in\mathbb{C}^{N_T\times1}$, where $N_T$ is the number of antennas.
At the receiver, we denote the voltages and currents at the receiving antennas as $\mathbf{y}\in\mathbb{C}^{N_R\times1}$ and $\mathbf{i}_y\in\mathbb{C}^{N_R\times1}$, where $N_R$ is the number of antennas, and the voltages and currents at the output of the \gls{milac} as $\mathbf{z}\in\mathbb{C}^{N_Z\times1}$ and $\mathbf{i}_z\in\mathbb{C}^{N_Z\times1}$, where $N_Z$ the number of \gls{rf} chains.
The received signal $\mathbf{z}$ is related to the transmitted signal $\mathbf{s}$ by
\begin{equation}
\mathbf{z}=\mathbf{G}\mathbf{H}\mathbf{F}\mathbf{s},
\end{equation}
where $\mathbf{G}\in\mathbb{C}^{N_Z\times N_R}$ is the combining matrix implemented by the receiver-side \gls{milac}, $\mathbf{H}\in\mathbb{C}^{N_R\times N_T}$ is the wireless channel, and $\mathbf{F}\in\mathbb{C}^{N_T\times N_S}$ is the precoding matrix implemented by the transmitter-side \gls{milac}.
Our goal is to derive accurate expressions for $\mathbf{G}$, $\mathbf{H}$, and $\mathbf{F}$ that include the effects of mutual coupling between antennas.

\begin{figure}[t]
\centering
\includegraphics[width=0.48\textwidth]{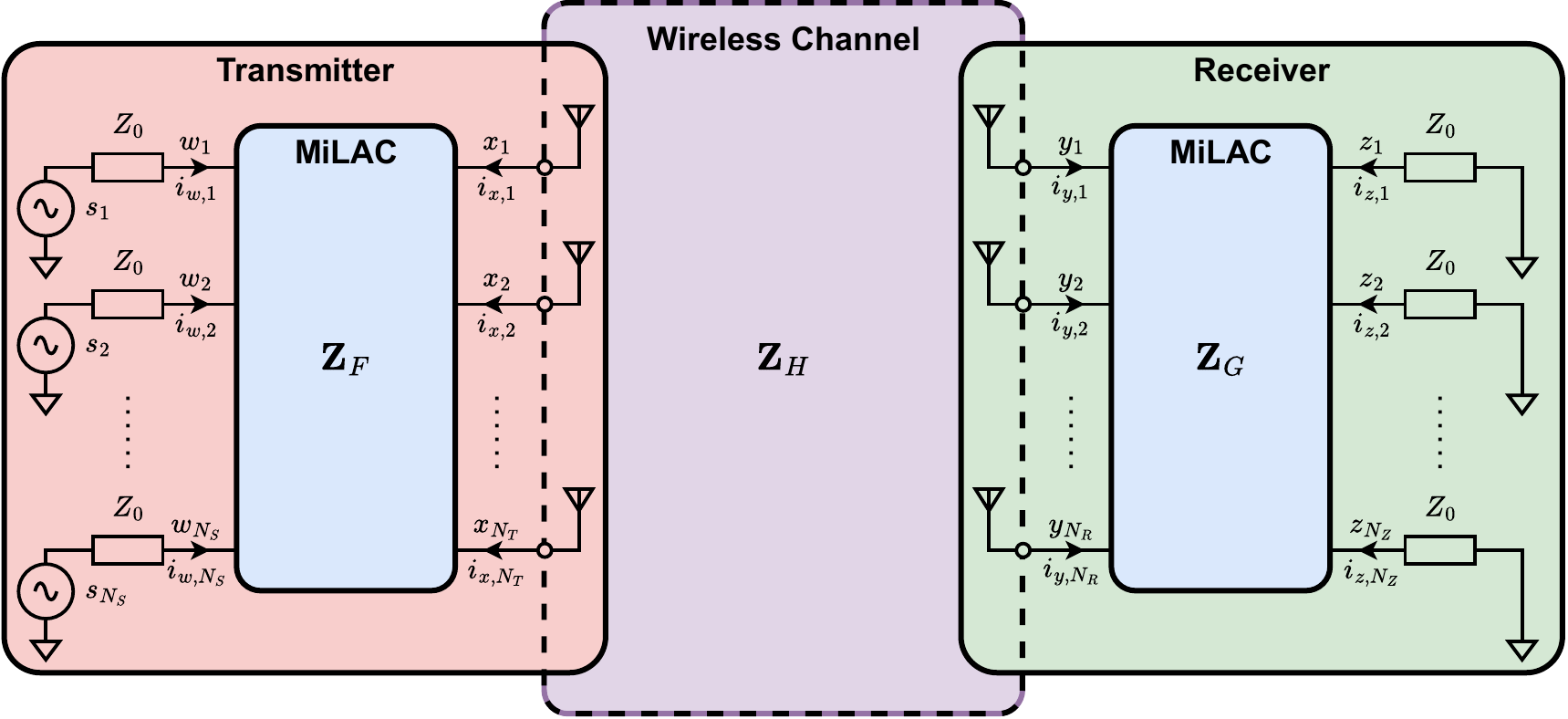}
\caption{Model of a MIMO system with MiLAC at both the transmitter and receiver.}
\label{fig:both}
\end{figure}

\subsection{Model With Mutual Coupling}

As done in the previous sections, we first derive the relationships between the signals at the wireless channel, the transmitter, and the receiver, and then solve the system of equations including all of them to obtain $\mathbf{G}$, $\mathbf{H}$, and $\mathbf{F}$.
The wireless channel is modeled as a $(N_T+N_R)$-port network with impedance matrix $\mathbf{Z}_H\in\mathbb{C}^{(N_T+N_R)\times(N_T+N_R)}$, partitioned as in \eqref{eq:ZH}.
Therefore, we have
\begin{equation}
\begin{bmatrix}
\mathbf{x}\\
\mathbf{y}
\end{bmatrix}=
\begin{bmatrix}
\mathbf{Z}_{TT} & \mathbf{Z}_{TR}\\
\mathbf{Z}_{RT} & \mathbf{Z}_{RR}
\end{bmatrix}
\begin{bmatrix}
-\mathbf{i}_x\\
-\mathbf{i}_y
\end{bmatrix},\label{eq:vZHi-both}
\end{equation}
where the currents are defined with directions as represented in Fig.~\ref{fig:both}.
At the transmitter, $\mathbf{w}$ and $\mathbf{i}_w$ are related to the source vector $\mathbf{s}$ as
\begin{equation}
\mathbf{w}=\mathbf{s}-Z_0\mathbf{i}_w,\label{eq:s-both}
\end{equation}
by Ohm's law, and, introducing the impedance matrix of the \gls{milac} as $\mathbf{Z}_F$, the voltages and currents at the \gls{milac} ports are related by \eqref{eq:vZFi}.
At the receiver, we similarly introduce the impedance matrix of the \gls{milac} as $\mathbf{Z}_{G}$ such that the relationship between voltages and currents in \eqref{eq:vZGi} holds.
In addtion, $\mathbf{z}$ and $\mathbf{i}_z$ are also related by
\begin{equation}
\mathbf{z}=-Z_0\mathbf{i}_z,\label{eq:z-both}
\end{equation}
because of Ohm's law.

The matrices $\mathbf{G}$, $\mathbf{H}$, and $\mathbf{F}$ can be obtained by solving the system of the five equations \eqref{eq:vZFi}, \eqref{eq:vZGi}, \eqref{eq:vZHi-both}, \eqref{eq:s-both}, and \eqref{eq:z-both}.
With the unilateral approximation \cite{ivr10}, i.e., $\mathbf{Z}_{TR}=\mathbf{0}$, it is possible to show that a procedure similar to the one in the previous sections gives
\begin{multline}
\mathbf{G}\mathbf{H}\mathbf{F}
=\left[\left(\frac{\mathbf{Y}_G}{Y_0}+
\begin{bmatrix}
\frac{\mathbf{Z}_{RR}^{-1}}{Y_0} & \mathbf{0}\\
\mathbf{0} & \mathbf{I}
\end{bmatrix}
\right)^{-1}\right]_{N_R+(1:N_Z),1:N_R}\\
\times Z_0\mathbf{Z}_{RR}^{-1}\mathbf{Z}_{RT}\mathbf{Z}_{TT}^{-1}\\
\times\left[\left(\frac{\mathbf{Y}_F}{Y_0}+
\begin{bmatrix}
\mathbf{I} & \mathbf{0}\\
\mathbf{0} & \frac{\mathbf{Z}_{TT}^{-1}}{Y_0}
\end{bmatrix}
\right)^{-1}\right]_{N_S+(1:N_T),1:N_S},
\end{multline}
allowing us to identify the channel $\mathbf{H}$ as
\begin{equation}
\mathbf{H}
=Z_0\mathbf{Z}_{RR}^{-1}\mathbf{Z}_{RT}\mathbf{Z}_{TT}^{-1},\label{eq:H-both}
\end{equation}
the precoder $\mathbf{F}$ as in \eqref{eq:F-tx}, and the combiner $\mathbf{G}$ as in \eqref{eq:G-rx}.

Departing from the system of five equations \eqref{eq:vZFi}, \eqref{eq:vZGi}, \eqref{eq:vZHi-both}, \eqref{eq:s-both}, and \eqref{eq:z-both}, and executing different computations, we can reach an alternative expression of the end-to-end channel $\mathbf{G}\mathbf{H}\mathbf{F}$ given by
\begin{equation}
\mathbf{G}\mathbf{H}\mathbf{F}=Z_0\left(\mathbf{Z}_{R}+Z_0\mathbf{I}\right)^{-1}\mathbf{J}_{R}\mathbf{Z}_{RT}\mathbf{J}_{T}^T\left(\mathbf{Z}_{T}+Z_0\mathbf{I}\right)^{-1},
\end{equation}
where
\begin{gather}
\mathbf{Z}_{T}=\mathbf{Z}_{F,11}-\mathbf{J}_{T}\mathbf{Z}_{F,21},\\
\mathbf{Z}_{R}=\mathbf{Z}_{G,22}-\mathbf{J}_{R}\mathbf{Z}_{G,12},\\
\mathbf{J}_{T}=\mathbf{Z}_{F,12}\left(\mathbf{Z}_{F,22}+\mathbf{Z}_{TT}\right)^{-1},\\
\mathbf{J}_{R}=\mathbf{Z}_{G,21}\left(\mathbf{Z}_{G,11}+\mathbf{Z}_{RR}\right)^{-1}.
\end{gather}
This expression is in exact agreement with the end-to-end channel model of a digital \gls{mimo} system with a matching network at both the transmitter and receiver, as derived in \cite[Section~II]{ivr10}.
However, \glspl{milac} can be reconfigured on a per-channel realization basis, unlike matching networks, which are fixed, and it is therefore more convenient to characterize their effect through their admittance matrices rather than impedance matrices.

\subsection{Model With No Mutual Coupling}

The derived system model can be simplified by assuming all antennas to be perfectly matched and neglecting the mutual coupling effects, yielding $\mathbf{Z}_T=Z_0\mathbf{I}$ and $\mathbf{Z}_R=Z_0\mathbf{I}$.
By substituting $\mathbf{Z}_T=Z_0\mathbf{I}$ and $\mathbf{Z}_R=Z_0\mathbf{I}$ in \eqref{eq:H-both}, the channel matrix simplifies as $\mathbf{H}=\mathbf{Z}_{RT}/Z_0$,
while substituting $\mathbf{Z}_T=Z_0\mathbf{I}$ and $\mathbf{Z}_R=Z_0\mathbf{I}$ in the precoder $\mathbf{F}$ and the combiner $\mathbf{G}$ given by \eqref{eq:F-tx} and \eqref{eq:G-rx}, they simplifies as in \eqref{eq:F-noMC-tx} and \eqref{eq:G-noMC-rx}, respectively.

\section{Physics-Compliant Optimization of MiLAC}
\label{sec:opt}

We have modeled \gls{milac}-aided \gls{mimo} systems accounting for the mutual coupling effects at the transmitter and receiver, and shown how the end-to-end system models vary depending on these effects.
In this section, we focus on a \gls{miso} system with \gls{milac} at the transmitter and optimize the \gls{milac} in the presence of mutual coupling to maximize the received signal power.

Consider a \gls{miso} system with \gls{milac} at the transmitter, namely the same as in Section~\ref{sec:tx} with $N_S=1$ \gls{rf} chain at the transmitter and $N_R=1$ antenna at the receiver, as shown in Fig.~\ref{fig:miso}(a).
Assuming the receiving antenna to be perfectly matched, i.e., $z_{RR}=Z_0$, the received signal model derived in Section~\ref{sec:tx} boils down to
\begin{equation}
z=\mathbf{h}\mathbf{f}s,
\end{equation}
where $s$ is the transmitted symbol such that $\mathbb{E}[\vert s\vert^2]=P_T$, with $P_T$ being the transmitted signal power.
The wireless channel $\mathbf{h}\in\mathbb{C}^{1\times N_T}$ is
\begin{equation}
\mathbf{h}
=\frac{1}{2}\mathbf{z}_{RT}\mathbf{Y}_{TT},\label{eq:h}
\end{equation}
and the precoder implemented by the \gls{milac} $\mathbf{f}\in\mathbb{C}^{N_T\times 1}$ is
\begin{equation}
\mathbf{f}=\left[\left(\frac{\mathbf{Y}}{Y_0}+
\begin{bmatrix}
1 & \mathbf{0}\\
\mathbf{0} & \frac{\mathbf{Y}_{TT}}{Y_0}
\end{bmatrix}\right)^{-1}\right]_{2:N_T+1,1},\label{eq:f}
\end{equation}
following \eqref{eq:H-tx} and \eqref{eq:F-tx}, where we have denoted the admittance matrix of the \gls{milac} as $\mathbf{Y}$ and introduced $\mathbf{Y}_{TT}=\mathbf{Z}_{TT}^{-1}$ to simplify the notation.
With perfect matching and no mutual coupling at the transmitter, we have $\mathbf{Y}_{TT}=Y_0\mathbf{I}$ and the expressions of $\mathbf{h}$ and $\mathbf{f}$ simplify to
\begin{equation}
\mathbf{h}=\frac{\mathbf{z}_{RT}}{2Z_0},\;
\mathbf{f}=\left[\left(\frac{\mathbf{Y}}{Y_0}+\mathbf{I}\right)^{-1}\right]_{2:N_T+1,1}.
\end{equation}

For a \gls{milac} that is lossless and reciprocal, the admittance matrix $\mathbf{Y}$ is purely imaginary and symmetric, i.e., it satisfies $\mathbf{Y}=j\mathbf{B}$ and $\mathbf{B}=\mathbf{B}^{T}$, where $\mathbf{B}\in\mathbb{R}^{(N_T+1)\times(N_T+1)}$ is the susceptance matrix of the \gls{milac}.
Our goal is therefore to optimize the admittance matrix $\mathbf{Y}$ subject to these constraints to maximize the received signal power $P=P_T\vert\mathbf{h}\mathbf{f}\vert^2$, with and without mutual coupling.

\begin{figure}[t]
\centering
\includegraphics[width=0.48\textwidth]{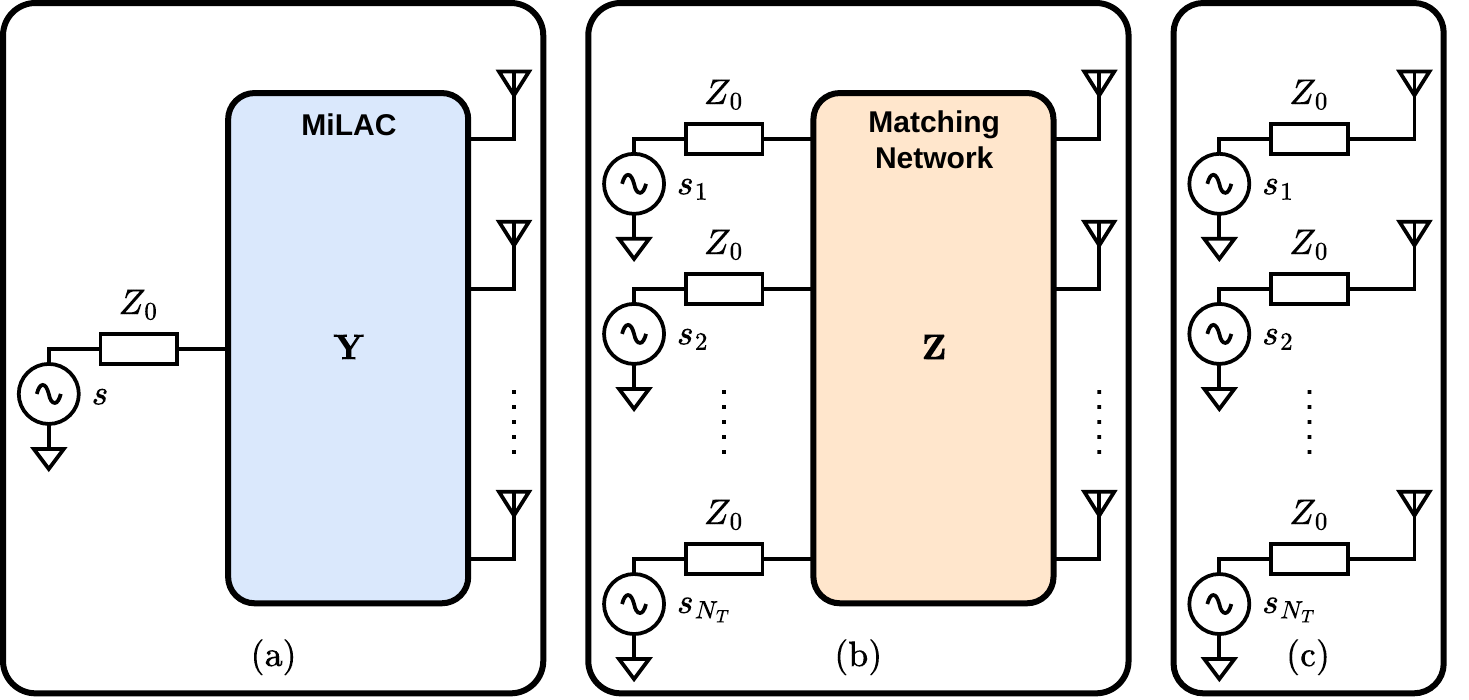}
\caption{Multi-antenna transmitter in a MISO system operating (a) MiLAC-aided beamforming, (b) digital beamforming with matching network, and (c) digital beamforming without matching network.}
\label{fig:miso}
\end{figure}

\subsection{Optimization of MiLAC With Mutual Coupling}

In the presence of mutual coupling, the received signal power maximization problem is given by
\begin{align}
\underset{\mathbf{B}}{\mathsf{\mathrm{max}}}\;\;
&P_T\left\vert\mathbf{h}\mathbf{f}\right\vert^2\label{eq:p1-o}\\
\mathsf{\mathrm{s.t.}}\;\;\;
&
\mathbf{f}=\left[\left(\frac{\mathbf{Y}}{Y_0}+
\begin{bmatrix}
1 & \mathbf{0}\\
\mathbf{0} & \frac{\mathbf{Y}_{TT}}{Y_0}
\end{bmatrix}\right)^{-1}\right]_{2:N_T+1,1},\\
&\mathbf{Y}=j\mathbf{B},\;
\mathbf{B}=\mathbf{B}^{T},\label{eq:p1-c}
\end{align}
where the optimization variable is the real symmetric matrix $\mathbf{B}$, which is the susceptance matrix of the \gls{milac}.
In the following, we derive in closed-form a solution to problem \eqref{eq:p1-o}-\eqref{eq:p1-c} that is proved to be globally optimal.

We begin by introducing the auxiliary constant terms $\hat{\mathbf{h}}\in\mathbb{C}^{1\times(N_T+1)}$ as $\hat{\mathbf{h}}=[0,\mathbf{h}]$, $\hat{\mathbf{Y}}_{TT}\in\mathbb{C}^{(N_T+1)\times(N_T+1)}$ as $\hat{\mathbf{Y}}_{TT}=\text{diag}(Y_0,\mathbf{Y}_{TT})$, and $\hat{\mathbf{e}}\in\mathbb{C}^{(N_T+1)\times 1}$ as $\hat{\mathbf{e}}=[1,0,\ldots,0]^T$.
Therefore, problem \eqref{eq:p1-o}-\eqref{eq:p1-c} can be equivalently rewritten as
\begin{align}
\underset{\mathbf{B}}{\mathsf{\mathrm{max}}}\;\;
&P_TY_0^2\left\vert\hat{\mathbf{h}}\left(j\mathbf{B}+\hat{\mathbf{Y}}_{TT}\right)^{-1}\hat{\mathbf{e}}\right\vert^{2}\label{eq:p2-o}\\
\mathsf{\mathrm{s.t.}}\;\;\;
&\mathbf{B}=\mathbf{B}^{T},\label{eq:p2-c}
\end{align}
which has the same form as the received signal power maximization problem in the presence of a fully-connected \gls{ris} with mutual coupling solved in \cite[Section~III]{ner26} through a global optimal solution.
In the following, we globally solve \eqref{eq:p2-o}-\eqref{eq:p2-c} following a similar solution to \cite[Section~III]{ner26}, inspired by the effect of the matching (or decoupling) network used in the context of \gls{ris} in \cite{sem26}.

We introduce the auxiliary variable $\bar{\mathbf{B}}\in\mathbb{R}^{(N_T+1)\times(N_T+1)}$ as a function of $\mathbf{B}$ as
\begin{equation}
\bar{\mathbf{B}}=Y_0\Re\{\hat{\mathbf{Y}}_{TT}\}^{-1/2}\left(\mathbf{B}+\Im\{\hat{\mathbf{Y}}_{TT}\}\right)\Re\{\hat{\mathbf{Y}}_{TT}\}^{-1/2},\label{eq:B-bar}
\end{equation}
which is a real matrix since $\mathbf{B}$, $\Im\{\hat{\mathbf{Y}}_{TT}\}$, and $\Re\{\hat{\mathbf{Y}}_{TT}\}^{-1/2}$ are real matrices.
Therefore, by substituting
\begin{equation}
\mathbf{B}=\frac{1}{Y_0}\Re\{\hat{\mathbf{Y}}_{TT}\}^{1/2}\bar{\mathbf{B}}\Re\{\hat{\mathbf{Y}}_{TT}\}^{1/2}-\Im\{\hat{\mathbf{Y}}_{TT}\},\label{eq:B}
\end{equation}
which follows from \eqref{eq:B-bar}, into \eqref{eq:p2-o}, problem \eqref{eq:p2-o}-\eqref{eq:p2-c} is equivalently rewritten as
\begin{align}
\underset{\mathbf{B}}{\mathsf{\mathrm{max}}}\;\;
&P_TY_0^2\left\vert\hat{\mathbf{h}}\Re\{\hat{\mathbf{Y}}_{TT}\}^{-1/2}\sqrt{Y_0}\right.\notag\\
&\left.\times\left(j\bar{\mathbf{B}}+Y_0\mathbf{I}\right)^{-1}\sqrt{Y_0}\Re\{\hat{\mathbf{Y}}_{TT}\}^{-1/2}\hat{\mathbf{e}}\right\vert^{2}\label{eq:p3-o}\\
\mathsf{\mathrm{s.t.}}\;\;\;
&\eqref{eq:B-bar},\;
\mathbf{B}=\mathbf{B}^{T}.\label{eq:p3-c}
\end{align}
Observe that constraint \eqref{eq:p3-c} means that $\bar{\mathbf{B}}$ can be an arbitrary symmetric matrix.
Therefore, problem \eqref{eq:p3-o}-\eqref{eq:p3-c} can be solved for a symmetric matrix $\bar{\mathbf{B}}$ and then $\mathbf{B}$ can be derived with \eqref{eq:B}.
To do so, we equivalently rewrite problem \eqref{eq:p3-o}-\eqref{eq:p3-c} as
\begin{align}
\underset{\bar{\mathbf{B}}}{\mathsf{\mathrm{max}}}\;\;
&P_TY_0^3\left\vert\hat{\mathbf{h}}\Re\{\hat{\mathbf{Y}}_{TT}\}^{-1/2}\left(j\bar{\mathbf{B}}+Y_0\mathbf{I}\right)^{-1}\hat{\mathbf{e}}\right\vert^{2}\label{eq:p4-o}\\
\mathsf{\mathrm{s.t.}}\;\;\;
&\bar{\mathbf{B}}=\bar{\mathbf{B}}^{T},\label{eq:p4-c}
\end{align}
where we also exploited the fact that $\sqrt{Y_0}\Re\{\hat{\mathbf{Y}}_{TT}\}^{-1/2}\hat{\mathbf{e}}=\hat{\mathbf{e}}$.

We now introduce another auxiliary variable $\bar{\mathbf{\Theta}}\in\mathbb{C}^{(N_T+1)\times(N_T+1)}$ as a function of $\bar{\mathbf{B}}$ as
\begin{equation}
\bar{\mathbf{\Theta}}=\left(Y_0\mathbf{I}+j\bar{\mathbf{B}}\right)^{-1}\left(Y_0\mathbf{I}-j\bar{\mathbf{B}}\right),\label{eq:T-bar}
\end{equation}
which helps in solving \eqref{eq:p4-o}-\eqref{eq:p4-c} because of the following two properties.
First, as a direct consequence of \eqref{eq:T-bar}, we have the relationship
$(j\bar{\mathbf{B}}+Y_0\mathbf{I})^{-1}=(\bar{\mathbf{\Theta}}+\mathbf{I})/(2Y_0)$, useful to simplify the objective function in \eqref{eq:p4-o}.
Second, since $\bar{\mathbf{B}}$ can be an arbitrary symmetric matrix, $\bar{\mathbf{\Theta}}$ can be an arbitrary unitary and symmetric matrix, from which $\bar{\mathbf{B}}$ can be recovered by inverting \eqref{eq:T-bar}.
These two properties allow us to rewrite \eqref{eq:p4-o}-\eqref{eq:p4-c} as
\begin{align}
\underset{\bar{\mathbf{\Theta}}}{\mathsf{\mathrm{max}}}\;\;
&\frac{P_TY_0}{4}\left\vert\hat{\mathbf{h}}\Re\{\hat{\mathbf{Y}}_{TT}\}^{-1/2}\left(\bar{\mathbf{\Theta}}+\mathbf{I}\right)\hat{\mathbf{e}}\right\vert^{2}\\
\mathsf{\mathrm{s.t.}}\;\;\;
&\bar{\mathbf{\Theta}}^H\bar{\mathbf{\Theta}}=\mathbf{I},\;\bar{\mathbf{\Theta}}=\bar{\mathbf{\Theta}}^{T},
\end{align}
which can be further simplified as
\begin{align}
\underset{\bar{\mathbf{\Theta}}}{\mathsf{\mathrm{max}}}\;\;
&\frac{P_TY_0}{4}\left\vert\mathbf{h}\Re\{\mathbf{Y}_{TT}\}^{-1/2}\left[\bar{\mathbf{\Theta}}\right]_{2:N_T+1,1}\right\vert^{2}\label{eq:p5-o1}\\
\mathsf{\mathrm{s.t.}}\;\;\;
&\bar{\mathbf{\Theta}}^H\bar{\mathbf{\Theta}}=\mathbf{I},\;\bar{\mathbf{\Theta}}=\bar{\mathbf{\Theta}}^{T},\label{eq:p5-c}
\end{align}
because of the definitions of $\hat{\mathbf{h}}$, $\hat{\mathbf{Y}}_{TT}$, and $\hat{\mathbf{e}}$.
As discussed in \cite{ner25-3}, the global optimal solution to this problem ensures that $\left[\bar{\mathbf{\Theta}}\right]_{2:N_T+1,1}=(\mathbf{h}\Re\{\mathbf{Y}_{TT}\}^{-1/2})^H/\Vert\mathbf{h}\Re\{\mathbf{Y}_{TT}\}^{-1/2}\Vert$, up to a phase shift.
A unitary symmetric matrix $\bar{\mathbf{\Theta}}$ fulfilling this condition is given as a function of the right singular vectors of $\mathbf{h}\Re\{\mathbf{Y}_{TT}\}^{-1/2}$, collected into the columns of a matrix denoted as $\mathbf{V}\in\mathbb{C}^{N_T\times N_T}$.
Partitioning $\mathbf{V}$ as $\mathbf{V}=[\bar{\mathbf{v}},\bar{\mathbf{V}}]$, with $\bar{\mathbf{v}}\in\mathbb{C}^{N_T\times 1}$ and $\bar{\mathbf{V}}\in\mathbb{C}^{N_T\times (N_T-1)}$, an optimal $\bar{\mathbf{\Theta}}$ is given by
\begin{equation}
\bar{\mathbf{\Theta}}=
\begin{bmatrix}
0 & \bar{\mathbf{v}}^T\\
\bar{\mathbf{v}} & \bar{\mathbf{V}}\bar{\mathbf{V}}^T
\end{bmatrix},
\end{equation}
which is unitary and symmetric by construction.

The maximum received signal power achieved by \gls{milac}-aided beamforming with mutual coupling is given by substituting $\left[\bar{\mathbf{\Theta}}\right]_{2:N_T+1,1}=(\mathbf{h}\Re\{\mathbf{Y}_{TT}\}^{-1/2})^H/\Vert\mathbf{h}\Re\{\mathbf{Y}_{TT}\}^{-1/2}\Vert$ into \eqref{eq:p5-o1}, giving
\begin{align}
P_{\text{MC}}^{\text{MiLAC}}
&=\frac{P_TY_0}{4}\left\Vert\mathbf{h}\Re\{\mathbf{Y}_{TT}\}^{-1/2}\right\Vert^2\\
&=\frac{P_TY_0}{16}\mathbf{z}_{RT}\mathbf{Y}_{TT}\Re\{\mathbf{Y}_{TT}\}\mathbf{Y}_{TT}^H\mathbf{z}_{RT}^H\label{eq:P-mc2},
\end{align}
where \eqref{eq:P-mc2} follows from $\mathbf{h}=\mathbf{z}_{RT}\mathbf{Y}_{TT}/2$ and since $\Vert\mathbf{v}\Vert^2=\mathbf{v}\mathbf{v}^H$ for any row vector $\mathbf{v}$.
We now observe that
\begin{multline}
\mathbf{Y}_{TT}\Re\{\mathbf{Y}_{TT}\}^{-1}\mathbf{Y}_{TT}^H=\\
\Re\{\mathbf{Y}_{TT}\}+\Im\{\mathbf{Y}_{TT}\}\Re\{\mathbf{Y}_{TT}\}^{-1}\Im\{\mathbf{Y}_{TT}\},
\end{multline}
by using $\mathbf{Y}_{TT}=\Re\{\mathbf{Y}_{TT}\}+j\Im\{\mathbf{Y}_{TT}\}$, which can be more concisely rewritten as
\begin{equation}
\mathbf{Y}_{TT}\Re\{\mathbf{Y}_{TT}\}^{-1}\mathbf{Y}_{TT}^H=\Re\{\mathbf{Z}_{TT}\}^{-1},\label{eq:YYY}
\end{equation}
by exploiting $(\mathbf{A}+j\mathbf{B})^{-1}=(\mathbf{A}+\mathbf{B}\mathbf{A}^{-1}\mathbf{B})^{-1}+j\mathbf{A}^{-1}\mathbf{B}(\mathbf{A}+\mathbf{B}\mathbf{A}^{-1}\mathbf{B})^{-1}$, with $\mathbf{A}$ and $\mathbf{B}$ being square real matrices \cite{dai24}.
Therefore, substituting \eqref{eq:YYY} into \eqref{eq:P-mc2}, we obtain
\begin{align}
P_{\text{MC}}^{\text{MiLAC}}
&=\frac{P_TY_0}{16}\mathbf{z}_{RT}\Re\{\mathbf{Z}_{TT}\}^{-1}\mathbf{z}_{RT}^H\label{eq:P-mc3}\\
&=\frac{P_TY_0}{16}\left\Vert\mathbf{z}_{RT}\Re\{\mathbf{Z}_{TT}\}^{-1/2}\right\Vert^2,\label{eq:P-mc4}
\end{align}
as the received signal power achievable by a \gls{milac} in a \gls{miso} system with mutual coupling.

We want now to derive the average performance $\mathbb{E}[P_{\text{MC}}^{\text{MiLAC}}]$ under the assumption that $\mathbf{z}_{RT}$ is a random variable with covariance matrix $\mathbb{E}[\mathbf{z}_{RT}^H\mathbf{z}_{RT}]=\rho\mathbf{I}$, i.e., in the presence of uncorrelated fading with path gain $\rho$.
By taking the expectation of \eqref{eq:P-mc3}, we obtain
\begin{align}
\mathbb{E}\left[P_{\text{MC}}^{\text{MiLAC}}\right]
&=\frac{P_TY_0}{16}\mathbb{E}\left[\mathbf{z}_{RT}\Re\{\mathbf{Z}_{TT}\}^{-1}\mathbf{z}_{RT}^H\right]\\
&=\frac{P_TY_0\rho}{16}\text{Tr}\left(\Re\{\mathbf{Z}_{TT}\}^{-1}\right),\label{eq:EP-mc3}
\end{align}
where \eqref{eq:EP-mc3} follows from the symmetry of the Frobenius inner product, the linearity of the trace, and the assumption that $\mathbb{E}[\mathbf{z}_{RT}^H\mathbf{z}_{RT}]=\rho\mathbf{I}$.
Interestingly, the average received signal power scales only with the trace term $\text{Tr}\left(\Re\{\mathbf{Z}_{TT}\}^{-1}\right)$, and depends only on the real part of the mutual coupling matrix.

\begin{table*}[t]
\centering
\caption{Comparison between MiLAC-aided and digital MISO systems.}
\begin{tabular}{|c|c|c|c|}
\hline
 & MiLAC & Digital with matching network & Digital without matching network\\
\hline
\# of RF chains &
$1$ &
$N_T$ &
$N_T$
\\
DACs resolution &
Low  &
High &
High
\\
\# of impedance components &
$\mathcal{O}(N_T)$ &
$\mathcal{O}(N_T^2)$ &
$0$
\\
Type of impedance components &
Tunable &
Fixed &
None
\\
Operations at each symbol time &
$0$ &
$\mathcal{O}(N_T)$ &
$\mathcal{O}(N_T)$
\\
\hline
$P$ with mutual coupling &
$\frac{P_TY_0}{16}\left\Vert\mathbf{z}_{RT}\Re\{\mathbf{Z}_{TT}\}^{-1/2}\right\Vert^2$ &
$\frac{P_TY_0}{16}\left\Vert\mathbf{z}_{RT}\Re\{\mathbf{Z}_{TT}\}^{-1/2}\right\Vert^2$ &
$\frac{P_T}{4}\left\Vert\mathbf{z}_{RT}\left(\mathbf{Z}_{TT}+Z_0\mathbf{I}\right)^{-1}\right\Vert^2$
\\
$\mathbb{E}[P]$ with mutual coupling &
$\frac{P_TY_0\rho}{16}\text{Tr}\left(\Re\{\mathbf{Z}_{TT}\}^{-1}\right)$ &
$\frac{P_TY_0\rho}{16}\text{Tr}\left(\Re\{\mathbf{Z}_{TT}\}^{-1}\right)$ &
$\frac{P_T\rho}{4}\text{Tr}\left(\left(\left(\mathbf{Z}_{TT}+Z_0\mathbf{I}\right)^H\left(\mathbf{Z}_{TT}+Z_0\mathbf{I}\right)\right)^{-1}\right)$
\\
$P$ without mutual coupling &
$\frac{P_TY_0^2}{16}\left\Vert\mathbf{z}_{RT}\right\Vert^2$ &
$\frac{P_TY_0^2}{16}\left\Vert\mathbf{z}_{RT}\right\Vert^2$ &
$\frac{P_TY_0^2}{16}\left\Vert\mathbf{z}_{RT}\right\Vert^2$
\\
$\mathbb{E}[P]$ without mutual coupling &
$\frac{P_TY_0^2\rho}{16}N_T$ &
$\frac{P_TY_0^2\rho}{16}N_T$ &
$\frac{P_TY_0^2\rho}{16}N_T$
\\
\hline
\end{tabular}
\label{tab}
\end{table*}

\subsection{Optimization of MiLAC With No Mutual Coupling}
\label{sec:opt-unaware}

We now show how the proposed optimization and performance analysis simplifies under the assumptions of perfect matching at all transmitting antennas and no mutual coupling, i.e., $\mathbf{Z}_{TT}=Z_0\mathbf{I}$, or, equivalently, $\mathbf{Y}_{TT}=Y_0\mathbf{I}$.
In this case, the received signal power maximization problem is
\begin{align}
\underset{\mathbf{B}}{\mathsf{\mathrm{max}}}\;\;
&P_T\left\vert\mathbf{h}\mathbf{f}\right\vert^{2}\label{eq:p1-o-no}\\
\mathsf{\mathrm{s.t.}}\;\;\;
&\mathbf{f}=\left[\left(\frac{\mathbf{Y}}{Y_0}+\mathbf{I}\right)^{-1}\right]_{2:N_T+1,1},\\
&\mathbf{Y}=j\mathbf{B},\;\mathbf{B}=\mathbf{B}^{T},\label{eq:p1-c-no}
\end{align}
where $\mathbf{B}$ is the susceptance of the \gls{milac}, which is real and symmetric.

A globally optimal solution to \eqref{eq:p1-o-no}-\eqref{eq:p1-c-no} can be found by introducing the scattering matrix of the \gls{milac} $\mathbf{\Theta}\in\mathbb{C}^{(N_T+1)\times(N_T+1)}$ as a function of $\mathbf{B}$ as
\begin{equation}
\mathbf{\Theta}=\left(Y_0\mathbf{I}+j\mathbf{B}\right)^{-1}\left(Y_0\mathbf{I}-j\mathbf{B}\right),\label{eq:T}
\end{equation}
which is unitary and symmetric.
By exploiting the relationship $(j\mathbf{B}/Y_0+\mathbf{I})^{-1}=(\mathbf{\Theta}+\mathbf{I})/2$, \eqref{eq:p1-o-no}-\eqref{eq:p1-c-no} can be rewritten as
\begin{align}
\underset{\mathbf{\Theta}}{\mathsf{\mathrm{max}}}\;\;
&\frac{P_T}{4}\left\vert\mathbf{h}\left[\mathbf{\Theta}\right]_{2:N_T+1,1}\right\vert^2\label{eq:p2-o-no}\\
\mathsf{\mathrm{s.t.}}\;\;\;
&\mathbf{\Theta}^H\mathbf{\Theta}=\mathbf{I},\;\mathbf{\Theta}=\mathbf{\Theta}^{T},\label{eq:p2-c-no}
\end{align}
which has a global optimal solution as discussed in \cite{ner25-3}.
Given the matrix $\mathbf{U}\in\mathbb{C}^{N_T\times N_T}$ containing the right singular vectors of $\mathbf{h}$ in its columns, partitioned as $\mathbf{U}=[\bar{\mathbf{u}},\bar{\mathbf{U}}]$, with $\bar{\mathbf{u}}\in\mathbb{C}^{N_T\times 1}$ and $\bar{\mathbf{U}}\in\mathbb{C}^{N_T\times (N_T-1)}$, it is possible to show that an optimal $\mathbf{\Theta}$ is
\begin{equation}
\mathbf{\Theta}=
\begin{bmatrix}
0 & \bar{\mathbf{u}}^T\\
\bar{\mathbf{u}} & \bar{\mathbf{U}}\bar{\mathbf{U}}^T
\end{bmatrix},
\end{equation}
which ensures that $[\mathbf{\Theta}]_{2:N_T+1,1}=\mathbf{h}^H/\Vert\mathbf{h}\Vert$.

The maximum received signal power achievable by a \gls{milac} with no mutual coupling can be readily obtained by substituting $[\mathbf{\Theta}]_{2:N_T+1,1}=\mathbf{h}^H/\Vert\mathbf{h}\Vert$ into \eqref{eq:p2-o-no}, which gives
\begin{align}
P_{\text{NoMC}}^{\text{MiLAC}}
&=\frac{P_T}{4}\left\Vert\mathbf{h}\right\Vert^2\\
&=\frac{P_TY_0^2}{16}\left\Vert\mathbf{z}_{RT}\right\Vert^2,\label{eq:P-noMC}
\end{align}
where \eqref{eq:P-noMC} follows from $\mathbf{h}=\mathbf{z}_{RT}/(2Z_0)$.
Note that $P_{\text{NoMC}}^{\text{MiLAC}}$ can also be obtained by substituting $\mathbf{Z}_{TT}=Z_0\mathbf{I}$ into the expression of $P_{\text{MC}}^{\text{MiLAC}}$ in \eqref{eq:P-mc4}, as expected.

We now derive the average performance $\mathbb{E}[P_{\text{NoMC}}^{\text{MiLAC}}]$ under the assumption that $\mathbf{z}_{RT}$ is a random variable with covariance matrix $\mathbb{E}[\mathbf{z}_{RT}^H\mathbf{z}_{RT}]=\rho\mathbf{I}$.
Taking the expectation of \eqref{eq:P-noMC}, we have
\begin{align}
\mathbb{E}\left[P_{\text{NoMC}}^{\text{MiLAC}}\right]
&=\frac{P_TY_0^2}{16}\mathbb{E}\left[\mathbf{z}_{RT}\mathbf{z}_{RT}^H\right]\\
&=\frac{P_TY_0^2\rho}{16}N_T,\label{eq:EP-nomc}
\end{align}
following the symmetry of the Frobenius inner product, the linearity of the trace, and applying our assumption $\mathbb{E}[\mathbf{z}_{RT}^H\mathbf{z}_{RT}]=\rho\mathbf{I}$.
Observe that the received signal power scales with the number of antennas $N_T$, as for a digital \gls{miso} system operating \gls{mrt}.

To analyze the impact of mutual coupling on \gls{milac}-aided systems, we present the following result comparing the average performance of \gls{milac} with and without mutual coupling.
\begin{proposition}
With uncorrelated fading, i.e., when $\mathbb{E}[\mathbf{z}_{RT}^H\mathbf{z}_{RT}]=\rho\mathbf{I}$, mutual coupling between the MiLAC antennas improves the average received signal power, i.e., 
\begin{equation}
\mathbb{E}\left[P_{\text{MC}}^{\text{MiLAC}}\right]\geq\mathbb{E}\left[P_{\text{NoMC}}^{\text{MiLAC}}\right].
\end{equation}
\label{pro:MiLAC}
\end{proposition}
\begin{proof}
Recalling that $\mathbb{E}[P_{\text{MC}}^{\text{MiLAC}}]$ and $\mathbb{E}[P_{\text{NoMC}}^{\text{MiLAC}}]$ are given by \eqref{eq:EP-mc3} and \eqref{eq:EP-nomc}, proving this proposition requires to show that $\text{Tr}(\Re\{\mathbf{Z}_{TT}\}^{-1})\geq Y_0N_T$.
This directly follows from \cite[Lemma~1]{ner26}, which is applicable since $\Re\{\mathbf{Z}_{TT}\}$ is a positive definite matrix with diagonal elements $[\Re\{\mathbf{Z}_{TT}\}]_{n,n}=1/Y_0$, considering perfect antenna matching.
\end{proof}
Interestingly, \gls{milac} can constructively exploit mutual coupling to enhance the received signal power.
It can be interpreted as a reconfigurable matching network that adapts to the channel realization, thereby performing beamforming while exploiting the effects of mutual coupling.

\section{Comparison with Digital Beamforming}
\label{sec:comparison}

Considering a \gls{milac}-aided \gls{miso} system with mutual coupling, we have derived a global optimal solution for the \gls{milac} and characterized its performance limits in closed form.
In this section, we compare the performance of \gls{milac}-aided beamforming with digital beamforming, considering the two cases when the digital transmitter is equipped with a matching network or not, as shown in Fig.~\ref{fig:miso}(b) and (c), respectively.

\subsection{Digital Beamforming With Matching Network}

The considered \gls{miso} system with \gls{milac} at the transmitter in Fig.~\ref{fig:miso}(a) can be compared with an equivalent \gls{miso} system operating digital beamforming with a matching network, as shown in Fig.~\ref{fig:miso}(b).
Following the system model derived in Section~\ref{sec:tx} with $N_R=1$ antenna at the receiver, which is assumed to be perfectly matched, i.e., $z_{RR}=Z_0$, we have that the received signal $z$ writes as
\begin{equation}
z=\mathbf{h}\mathbf{s},
\end{equation}
where $\mathbf{s}$ is the transmitted signal such that $\mathbb{E}[\Vert\mathbf{s}\Vert^2]=P_T$, with $P_T$ being the transmitted signal power, and $\mathbf{h}$ is given as a function of the matching network impedance matrix $\mathbf{Z}_F$ by 
\begin{equation}
\mathbf{h}=\frac{1}{2}\mathbf{z}_{RT}\mathbf{J}_{T}^T\left(\mathbf{Z}_{T}+Z_0\mathbf{I}\right)^{-1},
\end{equation}
where $\mathbf{Z}_{T}$ and $\mathbf{J}_{T}$ are defined in \eqref{eq:ZT} and \eqref{eq:JT}, respectively.
According to \cite[Section~V]{ivr10}, it is convenient to fix the impedance matrix of the matching network to
\begin{equation}
\mathbf{Z}_F=
\begin{bmatrix}
\mathbf{0} & -j\sqrt{Z_0}\Re\{\mathbf{Z}_{TT}\}^{1/2}\\
-j\sqrt{Z_0}\Re\{\mathbf{Z}_{TT}\}^{1/2} & -j\Im\{\mathbf{Z}_{TT}\}
\end{bmatrix},\label{eq:Zmatching}
\end{equation}
as a function of the mutual coupling matrix $\mathbf{Z}_{TT}$, such that all the power is delivered from the generators to the antennas.
Substituting \eqref{eq:Zmatching} into \eqref{eq:ZT} and \eqref{eq:JT}, we obtain $\mathbf{Z}_{T}=Z_0\mathbf{I}$ and $\mathbf{J}_{T}=-j\sqrt{Z_0}\Re\{\mathbf{Z}_{TT}\}^{-1/2}$, yielding
\begin{equation}
\mathbf{h}
=-\frac{j}{4\sqrt{Z_0}}\mathbf{z}_{RT}\Re\{\mathbf{Z}_{TT}\}^{-1/2}.\label{eq:h-mn}
\end{equation}
The received power is maximized with \gls{mrt}, which gives
\begin{align}
P_{\text{MC}}^{\text{Matching}}
&=P_T\left\Vert\mathbf{h}\right\Vert^2\\
&=\frac{P_TY_0}{16}\left\Vert\mathbf{z}_{RT}\Re\{\mathbf{Z}_{TT}\}^{-1/2}\right\Vert^2,\label{eq:P-mn}
\end{align}
where \eqref{eq:P-mn} follows from \eqref{eq:h-mn}.

Comparing the performance of \gls{milac}-aided beamforming in \eqref{eq:P-mc4} with the performance of digital beamforming with a matching network, we readily obtain the following result.
\begin{proposition}
MiLAC achieves the same received power as digital beamforming with matching network for any channel realization $\mathbf{z}_{RT}$, i.e., $P_{\text{MC}}^{\text{MiLAC}}=P_{\text{MC}}^{\text{Matching}}$.
\label{pro:mn}
\end{proposition}
A detailed comparison between \gls{milac}-aided beamforming and digital beamforming with a matching network is reported in Table~\ref{tab}.
Following Proposition~\ref{pro:mn}, these two strategies achieve the same performance for any channel realization, both in the presence or in the absence of mutual coupling.
There are however differences in terms of hardware implementation.
A \gls{milac} only requires one \gls{rf} chain to perform single-stream transmission, which can include low-resolution \glspl{dac} since it only carries the transmitted symbol, lying in a constellation with finite cardinality.
The number of impedance components required by a \gls{milac} serving a single-antenna user is $2N_T+1$, scaling with $\mathcal{O}(N_T)$, because of the optimality of the stem-connected architecture proposed in \cite{ner25-4}, while the number of impedance components in a $2N_T$-port fully-connected matching network scales with $\mathcal{O}(N_T^2)$.
\gls{milac} also does not require any computation at each symbol time, while the price to pay is that its microwave network includes tunable components, which are more challenging to implement than fixed ones.

\subsection{Digital Beamforming Without Matching Network}

We now consider a digital \gls{miso} system without the matching network, as shown in Fig.~\ref{fig:miso}(c).
Following the system model derived in Section~\ref{sec:mimo} with $N_R=1$ antenna at the receiver, which is assumed to be perfectly matched, i.e., $z_{RR}=Z_0$, the received signal $z$ writes as
\begin{equation}
z=\mathbf{h}\mathbf{s},
\end{equation}
where $\mathbf{s}$ is the transmitted signal such that $\mathbb{E}[\Vert\mathbf{s}\Vert^2]=P_T$, with $P_T$ being the transmitted signal power, and $\mathbf{h}$ is the wireless channel given by 
\begin{equation}
\mathbf{h}
=\frac{1}{2}\mathbf{z}_{RT}\left(\mathbf{Z}_{TT}+Z_0\mathbf{I}\right)^{-1},\label{eq:h-digital}
\end{equation}
The received power is maximized with \gls{mrt}, which gives
\begin{align}
P_{\text{MC}}^{\text{Digital}}
&=P_T\left\Vert\mathbf{h}\right\Vert^2\\
&=\frac{P_T}{4}\left\Vert\mathbf{z}_{RT}\left(\mathbf{Z}_{TT}+Z_0\mathbf{I}\right)^{-1}\right\Vert^2.\label{eq:P-digital}
\end{align}
where \eqref{eq:P-digital} follows from \eqref{eq:h-digital}.

The following proposition compares the received signal power of \gls{milac} with that of digital beamforming with no matching network in the presence of mutual coupling.
\begin{proposition}
With mutual coupling, MiLAC achieves higher received signal power than digital beamforming with no matching network for any channel realization $\mathbf{z}_{RT}$, i.e., 
\begin{equation}
P_{\text{MC}}^{\text{MiLAC}}\geq P_{\text{MC}}^{\text{Digital}}.
\end{equation}
\label{pro:digital}
\end{proposition}
\begin{proof}
Recalling that $P_{\text{MC}}^{\text{MiLAC}}$ and $P_{\text{MC}}^{\text{Digital}}$ are given by \eqref{eq:P-mc4} and \eqref{eq:P-digital}, proving this proposition requires to show that
\begin{equation}
\frac{Y_0}{4}\left\Vert\mathbf{z}_{RT}\Re\{\mathbf{Z}_{TT}\}^{-1/2}\right\Vert^2\geq\left\Vert\mathbf{z}_{RT}\left(\mathbf{Z}_{TT}+Z_0\mathbf{I}\right)^{-1}\right\Vert^2,
\end{equation}
for any $\mathbf{z}_{RT}\in\mathbb{C}^{1\times N_T}$, which is the same as verifying the matrix inequality
\begin{equation}
\frac{Y_0}{4}\Re\{\mathbf{Z}_{TT}\}^{-1}\succcurlyeq\left(\left(\mathbf{Z}_{TT}+Z_0\mathbf{I}\right)^H\left(\mathbf{Z}_{TT}+Z_0\mathbf{I}\right)\right)^{-1}.
\end{equation}
Since for $\mathbf{A}$ and $\mathbf{B}$ positive definite, $\mathbf{A}\succcurlyeq\mathbf{B}$ is the same as $\mathbf{B}^{-1}\succcurlyeq\mathbf{A}^{-1}$, our condition becomes
\begin{equation}
\left(\mathbf{Z}_{TT}+Z_0\mathbf{I}\right)^H\left(\mathbf{Z}_{TT}+Z_0\mathbf{I}\right)\succcurlyeq4Z_0\Re\{\mathbf{Z}_{TT}\},
\end{equation}
which can be equivalently rewritten as
\begin{equation}
\mathbf{Z}_{TT}^H\mathbf{Z}_{TT}+2Z_0\Re\{\mathbf{Z}_{TT}\}+Z_0^2\mathbf{I}\succcurlyeq4Z_0\Re\{\mathbf{Z}_{TT}\}.
\end{equation}
Since saying $\mathbf{A}\succcurlyeq\mathbf{B}$ is the same as saying $\mathbf{A}-\mathbf{B}\succcurlyeq\mathbf{0}$, our condition becomes
\begin{equation}
\mathbf{Z}_{TT}^H\mathbf{Z}_{TT}-2Z_0\Re\{\mathbf{Z}_{TT}\}+Z_0^2\mathbf{I}\succcurlyeq\mathbf{0},
\end{equation}
which can be equivalently rewritten as
\begin{equation}
\left(\mathbf{Z}_{TT}-Z_0\mathbf{I}\right)^H\left(\mathbf{Z}_{TT}-Z_0\mathbf{I}\right)\succcurlyeq\mathbf{0},
\end{equation}
which is always satisfied since $\left(\mathbf{Z}_{TT}-Z_0\mathbf{I}\right)^H\left(\mathbf{Z}_{TT}-Z_0\mathbf{I}\right)$ is positive semi-definite.
\end{proof}

The equality in Proposition~\eqref{pro:digital} is satisfied when $\mathbf{Z}_{TT}=Z_0\mathbf{I}$, i.e., with perfect matching and no mutual coupling.
This implies that \gls{milac} performs as digital beamforming in a \gls{miso} system with perfect matching and no mutual coupling, in line with the results in \cite{ner25-3}.
Finally, the average performance $\mathbb{E}[P_{\text{MC}}^{\text{Digital}}]$ in the case $\mathbf{z}_{RT}$ has covariance matrix $\mathbb{E}[\mathbf{z}_{RT}^H\mathbf{z}_{RT}]=\rho\mathbf{I}$, i.e., in the presence of uncorrelated fading with path gain $\rho$, is readily given by taking the expectation of \eqref{eq:P-digital} as
\begin{equation}
\mathbb{E}\left[P_{\text{MC}}^{\text{Digital}}\right]
=\frac{P_T\rho}{4}\text{Tr}\left(\left(\left(\mathbf{Z}_{TT}+Z_0\mathbf{I}\right)^H\left(\mathbf{Z}_{TT}+Z_0\mathbf{I}\right)\right)^{-1}\right).
\end{equation}
The comprehensive comparison between \gls{milac}-aided and digital systems is summarized in Table~\ref{tab}, which shows that the only benefits of digital beamforming with no matching network is that it does not require any \gls{rf} component, and therefore reduced losses are expected in practical systems.

\section{Numerical Results}
\label{sec:results}

This section provides simulation results to validate the derived global optimal closed-form solutions and performance bounds for \gls{milac} with mutual coupling.
Consider a \gls{miso} system with \gls{milac} at the transmitter.
The antenna array of the \gls{milac} is a \gls{upa} of antennas located in the $x$-$y$ plane, with dimensions $N_X\times N_Y$, where $N_X=8$ and $N_Y=N_T/8$, and with antenna spacing $d$.
The antennas are thin wire dipoles parallel to the $y$ axis with length $\ell=\lambda/4$ and radius $r\ll\ell$, where $\lambda=c/f$ is the wavelength of the frequency $f=28$~GHz, and $c$ is the speed of light.
All the antennas are assumed to be perfectly matched to $Z_0=50\;\Omega$, giving $\left[\mathbf{Z}_{TT}\right]_{n,n}=Z_0$, for $n=1,\ldots,N_T$.
In addition, the $(q,p)$th entry of $\mathbf{Z}_{TT}$, with $q\neq p$, represent the mutual coupling between the antenna $p$ located in $(x_p,y_p)$ and the antenna $q$ located in $(x_q,y_q)$.
Following \cite{gra21}, the entry $\left[\mathbf{Z}_{TT}\right]_{q,p}=\left[\mathbf{Z}_{TT}\right]_{p,q}$ is modeled as
\begin{multline}
\left[\mathbf{Z}_{TT}\right]_{q,p}=
\int_{y_q-\frac{\ell}{2}}^{y_q+\frac{\ell}{2}}
\int_{y_p-\frac{\ell}{2}}^{y_p+\frac{\ell}{2}}
\frac{j\eta_0}{4\pi k_0}
\left(\frac{\left(y''-y'\right)^2}{d_{q,p}^2}\right.\\
\left.\times\left(\frac{3}{d_{q,p}^2}+\frac{3jk_0}{d_{q,p}}-k_0^2\right)
-\frac{jk_0+d_{q,p}^{-1}}{d_{q,p}}+k_0^2\right)
\frac{e^{-jk_0d_{q,p}}}{d_{q,p}}\\
\times\frac{
\sin\left(k_0\left(\frac{\ell}{2}-\left\vert y'-y_p\right\vert\right)\right)
\sin\left(k_0\left(\frac{\ell}{2}-\left\vert y''-y_q\right\vert\right)\right)}
{\sin^2\left(k_0\frac{\ell}{2}\right)}dy'dy'',
\end{multline}
where $\eta_0=377\;\Omega$ is the impedance of free space, $k_0=2\pi/\lambda$ is the wavenumber, and $d_{q,p}=\sqrt{(x_q-x_p)^2+(y''-y')^2}$.
We generate $\mathbf{z}_{RT}$ as independent Rayleigh distributed with unit path gain, i.e., $\mathbf{z}_{RI}\sim\mathcal{CN}(\mathbf{0},\mathbf{I})$.
We consider unit path gain and unit transmit signal power, i.e., $\rho=1$ and $P_T=1$, for simplicity, since they impact the received signal power by just scaling it by a constant factor.

\begin{figure}[t]
\centering
\includegraphics[height=0.27\textwidth]{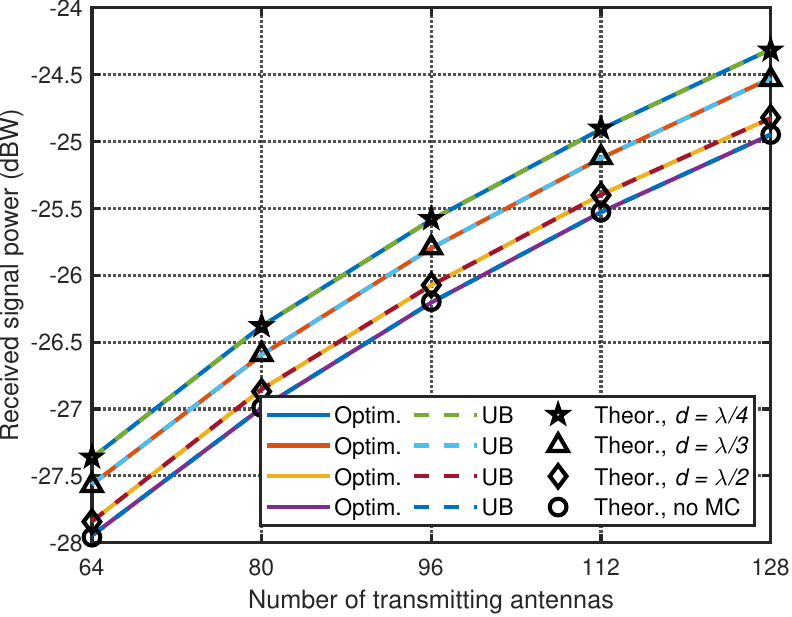}
\caption{Received signal power versus the number of transmitting antennas achieved by MiLAC.}
\label{fig:aware}
\end{figure}

In Fig.~\ref{fig:aware}, we report the received signal power obtained by \gls{milac} with mutual coupling for different values of antenna spacing $d\in[\lambda/2,\lambda/3,\lambda/4]$, and without mutual coupling, i.e., with $\mathbf{Z}_{TT}=Z_0\mathbf{I}$.
We compare the following three baselines.
\begin{itemize}
\item Optim.: The average received signal power resulting from Monte Carlo simulations, where at each channel realization the \gls{milac} is optimized with the proposed algorithm in Section~\ref{sec:opt}.
\item UB: The average received signal power resulting from Monte Carlo simulations, where at each channel realization the received signal power is given by its upper bound $P_{\text{MC}}^{\text{MiLAC}}=P_TY_0\Vert\mathbf{z}_{RT}\Re\{\mathbf{Z}_{TT}\}^{-1/2}\Vert^2/16$ or $P_{\text{NoMC}}^{\text{MiLAC}}=P_TY_0^2\Vert\mathbf{z}_{RT}\Vert^2/16$.
\item Theor.: The theoretical average received signal power given by $\mathbb{E}[P_{\text{MC}}^{\text{MiLAC}}]=P_TY_0\rho\text{Tr}(\Re\{\mathbf{Z}_{TT}\}^{-1})/16$ or $\mathbb{E}[P_{\text{NoMC}}^{\text{MiLAC}}]=P_TY_0^2\rho N_T/16$.
\end{itemize}

We make the following four observations from Fig.~\ref{fig:aware}.
\textit{First}, the \gls{milac} optimized with the proposed solution exactly achieves the received signal power upper bound, confirming the effectiveness of our global optimal closed-form solution.
\textit{Second}, the average received signal power obtained with Monte Carlo simulations exactly corresponds to the theoretical average derived in closed form, confirming the validity of our closed-form scaling laws.
\textit{Third}, the presence of mutual coupling allows the \gls{milac} to achieve higher received signal power than with no mutual coupling, confirming Proposition~\ref{pro:MiLAC}.
A similar trend was also analytically proved for \gls{ris}-aided systems in \cite{ner26}.
\textit{Fourth}, when $d=\lambda$, the mutual coupling is so weak that the achieved performance is approximately the same as in the absence of mutual coupling.

We extend the comparison to include systems where the \gls{milac} is optimized without accounting for mutual coupling.
In Fig.~\ref{fig:unaware}, we report the received signal power achieved by the \gls{milac}, when it is optimized through mutual coupling aware as well as unaware algorithms, for different numbers of antennas $N_T\in[64,96,128]$.
For \gls{milac} reconfigured with mutual coupling aware optimization, the performance is given by the received signal power upper bound $P_{\text{MC}}^{\text{MiLAC}}=P_TY_0\Vert\mathbf{z}_{RT}\Re\{\mathbf{Z}_{TT}\}^{-1/2}\Vert^2/16$ or $P_{\text{NoMC}}^{\text{MiLAC}}=P_TY_0^2\Vert\mathbf{z}_{RT}\Vert^2/16$.
For \gls{milac} optimized in a mutual coupling unaware fashion, we assume that $\mathbf{Z}_{II}=Z_0\mathbf{I}$ during the optimization phase, and optimize the \gls{milac} by using the solution proposed in Section~\ref{sec:opt-unaware}.
From the obtained $\boldsymbol{\Theta}$, the susceptance matrix $\mathbf{B}$ is obtained by inverting \eqref{eq:T}, and it is plugged into the channel model with mutual coupling in \eqref{eq:h}-\eqref{eq:f} to get the received signal power of a \gls{milac}-aided system with mutual coupling unaware optimization.

We make the following three remarks from Fig.~\ref{fig:unaware}.
\textit{First}, when optimizing the \gls{milac} accounting for mutual coupling, the performance increases with the mutual coupling strength, i.e., as the antenna spacing decreases.
This is because \gls{milac} can effectively deliver the power to the coupled antennas, as a matching network.
However, when the \gls{milac} is optimized not being aware of the mutual coupling, its performance dramatically drops for small values of antenna spacing, experiencing a degradation of up to 3~dB.
This is because the \gls{milac} is optimized based on a channel model with no mutual coupling, which differs from the physics-compliant channel model with the mutual coupling effects.
\textit{Second}, the impact of mutual coupling can be approximately neglected during the optimization when the antenna spacing is larger than half-wavelength, i.e., $d\geq\lambda/2$.
The \gls{milac} can be successfully optimized without considering mutual coupling in this case, given how weak the mutual coupling is.
\textit{Third}, the above remarks are independent of the number of antennas $N_T$.

\begin{figure}[t]
\centering
\includegraphics[height=0.27\textwidth]{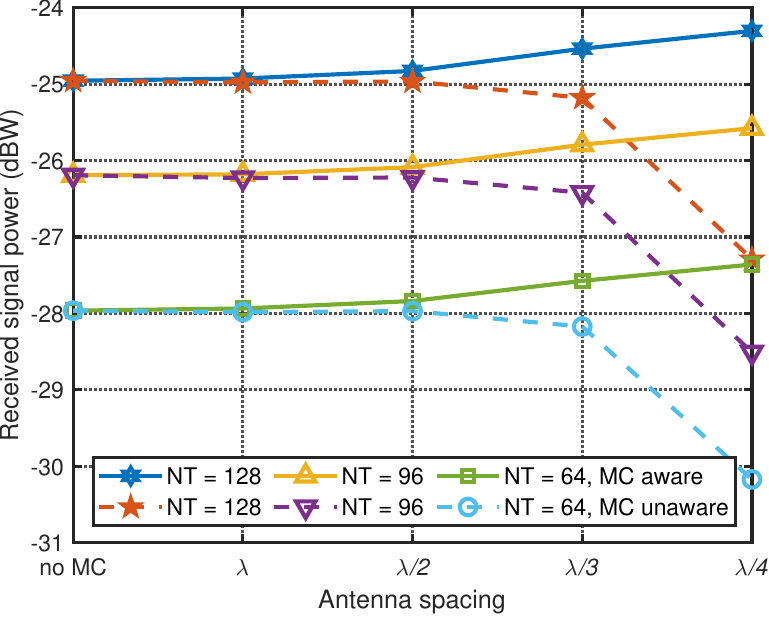}
\caption{Received signal power versus the antenna spacing achieved by MiLAC.}
\label{fig:unaware}
\end{figure}

\begin{figure}[t]
\centering
\includegraphics[height=0.27\textwidth]{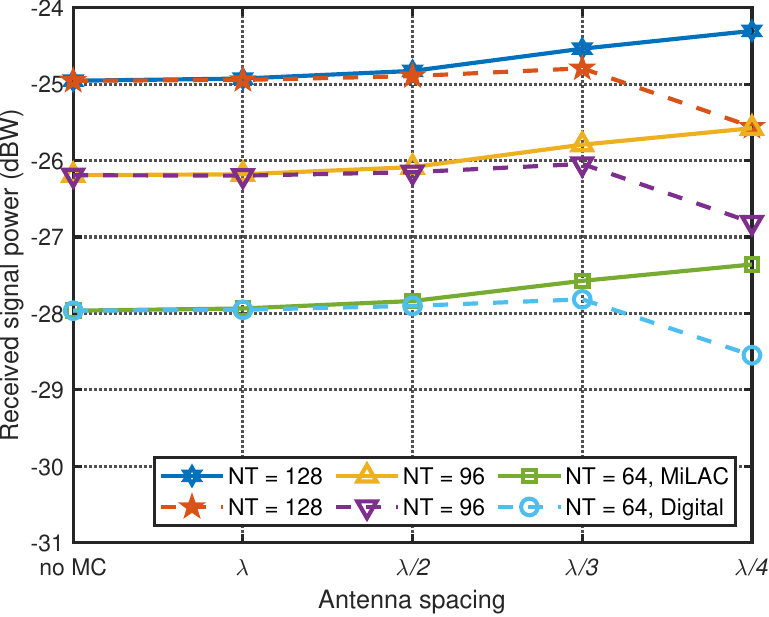}
\caption{Received signal power versus the antenna spacing achieved by MiLAC and digital beamforming with no matching network.}
\label{fig:digital}
\end{figure}

We finally compare \gls{milac} with digital beamforming with no matching network in the presence of mutual coupling.
In Fig.~\ref{fig:digital}, we report the received signal power achieved by \gls{milac} (optimized through the proposed globally optimal solution) and by digital beamforming (assuming \gls{mrt}, which is optimal in \gls{miso}).
For \gls{milac}, the performance is given by the received signal power upper bound $P_{\text{MC}}^{\text{MiLAC}}=P_TY_0\Vert\mathbf{z}_{RT}\Re\{\mathbf{Z}_{TT}\}^{-1/2}\Vert^2/16$ or $P_{\text{NoMC}}^{\text{MiLAC}}=P_TY_0^2\Vert\mathbf{z}_{RT}\Vert^2/16$.
For digital beamforming, the performance is given by the corresponding received signal power upper bound $P_{\text{MC}}^{\text{Digital}}=P_T\Vert\mathbf{z}_{RT}(\mathbf{Z}_{TT}+Z_0\mathbf{I})^{-1}\Vert^2/4$ or $P_{\text{NoMC}}^{\text{Digital}}=P_TY_0^2\Vert\mathbf{z}_{RT}\Vert^2/16$.
From Fig.~\ref{fig:digital}, we observe that \gls{milac}-aided beamforming outperforms digital beamforming in the presence of mutual coupling, validating Proposition~\ref{pro:digital}.
A \gls{milac} can outperform digital beamforming since it acts as a reconfigurable matching network which can simultaneously maximize the power flow from the \gls{rf} chain to the antennas and perform beamforming.
The discrepancy between the two beamforming strategies increases with the mutual coupling strength (i.e., as the antenna spacing decreases), while it vanishes in the absence of mutual coupling, in agreement with \cite{ner25-3}.

\section{Conclusion}

We have developed a rigorous physics-compliant modeling for \gls{milac}-aided \gls{mimo} systems by leveraging multiport network theory.
Specifically, we have derived end-to-end system models for conventional digital \gls{mimo} and \gls{milac}-aided systems, considering systems where the \gls{milac} is deployed at the transmitter, the receiver, or both.
The resulting models reveal how mutual coupling between the antennas fundamentally alters the channel and the precoding and combining matrices implemented by the \gls{milac}.

Building on the developed models, we have addressed the optimization of \glspl{milac} in the presence of mutual coupling.
For a \gls{milac}-aided \gls{miso} system, we have derived a closed-form solution that maximizes the received signal power under lossless and reciprocal constraints, and proved its global optimality.
We have further obtained closed-form expressions for the maximum achievable performance with and without mutual coupling, thereby showing that mutual coupling effects improve the performance in \gls{milac}-aided systems.
In addition, we have compared the performance of \gls{milac}-aided with digital \gls{mimo} systems, with or without a matching network.
Our analytical derivations show that \gls{milac} always performs as digital beamforming with a matching network, and always outperforms digital beamforming without a matching network.
Numerical results validate our theoretical findings and confirm that the proposed \gls{milac} optimization achieves the derived performance upper bounds, visibly outperforming digital \gls{mimo} systems when they are not equipped with a matching network.

\bibliographystyle{IEEEtran}
\bibliography{IEEEabrv,main}

\end{document}